\documentclass{LMCS}

\usepackage[nogroup]{versions}

\includeversionnogroup{LMCS}
\excludeversion{conference}
\excludeversion{web}

\usepackage{pwlogic,FO2_on_words}

\usepackage{caption}
\usepackage[basic]{complexity}
\usepackage{enumerate}
\usepackage{hyperref}
\usepackage{picins}
\usepackage{url}

\newcommand{\qedconf}{}

\processifversion{LMCS}{
  \includeversionnogroup{full}
}

\processifversion{web}{
  \includeversionnogroup{full}
  \usepackage[nocitecount]{drftcite}
  \usepackage{scrpage2}
  \pagestyle{scrheadings}
}

\processifversion{full}{
}

\processifversion{conference}{
  \renewcommand{\qedconf}{\qed}
  \newcommand{\qedhere}{\qedconf}
}

\newcommand{\abs}[1]{ \vert #1 \vert }
  
\newcommand{\commentout}[1]{}
\newcommand{\fo}{\FO}
\newcommand{\set}[1]{\{  #1 \} }

\usepackage{FO2_on_words_wordpics}

\def\doi{5 (3:3) 2009}
\lmcsheading%
{\doi}
{1--23}
{}
{}
{Feb.~\phantom{0}7, 2008}
{Aug.~\phantom{0}3, 2009}
{} 

\begin{document}

\author[P.~Weis]{Philipp Weis}
\author[N.~Immerman]{Neil Immerman}

\newcommand{\myaddress}{%
  Department of Computer Science\\
  University of Massachusetts, Amherst\\
  140 Governors Drive, Amherst, MA 01003, USA
}
\newcommand{\myinfo}{
  \myaddress\\
  \url{{pweis,immerman}@cs.umass.edu}\\
  \url{http://www.cs.umass.edu/~{pweis,immerman}}\\\bigskip\bigskip
  \footnotesize{1998 ACM Subject Classification: F.4.1, F.4.3}
}
\newcommand{\mythanks}{Supported in part by NSF grant CCF-0514621.}
\newcommand{\mytitle}{%
  Structure Theorem and Strict Alternation Hierarchy for
  \texorpdfstring{$\FOV{2}$}{FO^2} on Words\rsuper*} 

\begin{conference}
  \institute{\myinfo}
  \title{\mytitle\thanks{\mythanks}}
\end{conference}
\begin{LMCS}
  \address{\myaddress}
  \email{\{pweis,immerman\}@cs.umass.edu}
  \thanks{\mythanks}
  \title{\mytitle}
  \titlecomment{{\lsuper*}This work is an extended version of \cite{WI07}}
\end{LMCS}
\begin{web}
  \date{\myinfo}
  \title{\mytitle\thanks{\mythanks}}
\end{web}

\begin{abstract}
It is well-known that every first-order property on words is
expressible using at most three variables. The subclass of properties
expressible with only two variables is also quite interesting and
well-studied. We prove precise structure
theorems that characterize the exact expressive power of first-order
logic with two variables on words. Our results apply to both the case with
and without a successor relation.

For both languages, our structure theorems show exactly what is expressible
using a given quantifier depth, $n$, and using $m$ blocks of alternating
quantifiers, for any $m\leq n$. Using these characterizations, we prove,
among other results, that there is a strict hierarchy of alternating
quantifiers for both languages. The question whether there was such a
hierarchy had been completely open. As another consequence of our structural
results, we show that satisfiability for first-order logic with two
variables without successor, which is \NEXP-complete in general, becomes
\NP-complete once we only consider alphabets of a bounded size.
\end{abstract}

\maketitle

\section{Introduction} \label{sec:intro}

It is well-known that every first-order property on words is
expressible using at most three variables \cite{IK89,K68}. 
The subclass of properties expressible with only
two variables is also quite interesting and well-studied (Fact
\ref{fact}).  

In this paper we prove precise structure theorems that characterize
the exact expressive power of first-order logic with two variables on
words.  Our results apply to $\FOV[$<$]{2}$ and $\FOV{2}[<,\suc]$, the
latter of which includes the binary successor relation in addition to
the linear ordering on string positions.  

For both languages, our structure theorems show exactly what is expressible
using a given quantifier depth, $n$, and using $m$ blocks of alternating
quantifiers, for any $m\leq n$. Using these characterizations, we prove that
there is a strict hierarchy of alternating quantifiers for both languages.
The question whether there was such a hierarchy had been completely open
since it was asked in \cite{EVW97,EVW02}. As another consequence of our
structural results, we show that satisfiability for $\FOV{2}[<]$, which is
\NEXP-complete in general \cite{EVW02}, becomes \NP-complete once we only
consider alphabets of a bounded size.

Our motivation for studying $\FOV{2}$ on words comes from the desire to
understand the trade-off between formula size and number of
variables.  This is of great interest because, as is well-known, this
is equivalent to the trade-off between parallel time and number of
processors \cite{I99}.  Adler and Immerman \cite{AI03}
 introduced a game that can be used to determine the
minimum size of first-order formulas with a given number of variables
needed to express a given property.  These games, which are closely
related to the communication complexity games of Karchmer and
Wigderson  \cite{KW90}, were used to prove two optimal size bounds
for temporal logics \cite{AI03}.  Later Grohe and Schweikardt used
similar methods to study the size versus variable trade-off for
first-order logic on unary words \cite{GS05}.  They proved that all
first-order expressible properties of unary words are already
expressible with two variables and that the variable-size trade-off
between two versus three variables is polynomial whereas the trade-off
between three versus four variables is exponential.  They left open
the trade-off between $k$ and $k+1$ variables for $k\geq 4$.  While we
do not directly address that question here, our classification of
$\FOV{2}$ on words is a step towards the general understanding of the
expressive power of $\FO{}$ needed for progress on such trade-offs.

Our characterization of $\FOV{2}[<]$ and $\FOV{2}[<,\suc]$ on words is based
on the very natural notion of $n$-ranker (Definition \ref{def:n-ranker}).
Informally, a ranker is the position of a certain combination of letters in
a word. For example, $\rright_\mathtt{a}$ and $\rleft_\mathtt{b}$ are
1-rankers where $\rright_\mathtt{a}(w)$ is the position of the first
\texttt{a} in $w$ (from the left) and $\rleft_\mathtt{b}(w)$ is the position
of the first \texttt{b} in $w$ from the right. Similarly, the 2-ranker $r_2
= \rright_\mathtt{a} \rright_\mathtt{c}$ denotes the position of the first
\texttt{c} to the right of the first \texttt{a}, and the 3-ranker, $r_3 =
\rright_\mathtt{a} \rright_\mathtt{c} \rleft_\mathtt{b}$ denotes the
position of the first \texttt{b} to the left of $r_2$. If there is no such
letter then the ranker is undefined. For example, $r_3(\mathtt{cababcba})=
5$ and $r_3(\mathtt{acbbca})$ is undefined.

Our first structure theorem (Theorem \ref{thm:ranker-char}) says that the
properties expressible in $\FOVD[$<$]{2}{n}$, i.e. first-order logic with two
variables and quantifier depth $n$, are exactly boolean combinations of
statements of the form, ``$r$ is defined'', and ``$r$ is to the left (right)
of $r'$'' for $k$-rankers, $r$, and $k'$-rankers, $r'$, with $k\leq n$ and
$k' < n$. A non-quantitative version of this theorem was previously known
\cite{STV01}.\footnote{See item \ref{turtle} in Fact \ref{fact}: a ``turtle
  language'' is a language of the form ``$r$ is defined'', for some ranker,
  $r$.} Furthermore, a quantitative version in terms of iterated block
products of the variety of semi-lattices is presented in \cite{TT07}, based
on work by Straubing and Th\'erien \cite{ST02}.

Surprisingly, Theorem \ref{thm:ranker-char} can be generalized in almost
exactly the same form to characterize $\FOV{2}_{m,n}[<]$ where there are at
most $m$ blocks of alternating quantifiers, $m\leq n$. This second structure
theorem (Theorem \ref{thm:alt-ranker-char}) uses the notion of
$(m,n)$-ranker where there are $m$ blocks of $\rright$'s or $\rleft$'s, that
is, changing direction in rankers corresponds exactly to alternation of
quantifiers. Using Theorem \ref{thm:alt-ranker-char} we prove that there is
a strict alternation hierarchy for $\FOV{2}_n[<]$ (Theorem
\ref{thm:alt-hierarchy}) but that exactly at most $\abs{\Sigma} +1$
alternations are useful, where $\abs{\Sigma}$ is the size of the alphabet
(Theorem \ref{thm:alt-alphabet}).

The language $\FOV{2}[<,\suc]$ is more expressive than $\FOV{2}[<]$ because
it allows us to talk about consecutive strings of symbols\footnote{With
  three variables we can express $\suc(x,y)$ using the ordering: $x < y
  \land \forall z (z \le x \lor y \le z)$.}. For $\FOV{2}[<,\suc]$, a
straightforward generalization of $n$-ranker to $n$-successor-ranker allows
us to prove exact analogs of Theorems \ref{thm:ranker-char} and
\ref{thm:alt-ranker-char}. We use the latter to prove that there is also a
strict alternation hierarchy for $\FOV{2}_n[<,\suc]$ (Theorem
\ref{thm:suc-alt-hierarchy}). Since in the presence of successor we can
encode an arbitrary alphabet in binary, no analog of Theorem
\ref{thm:alt-alphabet} holds for $\FOV{2}[<,\suc]$.

The expressive power of first-order logic with three or more variables on
words has been well-studied. The languages expressible are of course the
star-free regular languages \cite{MP71}. The dot-depth hierarchy is the
natural hierarchy of these languages. This hierarchy is strict \cite{BK78}
and identical to the first-order quantifier alternation hierarchy
\cite{T82,T84}.

Many beautiful results on $\FOV{2}$ on words were also already known. The
main significant outstanding question was whether there was an alternation
hierarchy. The following is a summary of the main previously known
characterizations of $\FOV{2}[<]$ on words. For a detailed treatment of all
these characterizations, we refer the reader to \cite{TT01}.

\begin{fact}\label{fact}
  \textup{\cite{EVW97,EVW02,PW97,S76,TW98,STV01}}
  Let $R\subseteq \Sigma^\star$. The following conditions are equivalent:
  \setlength{\itemsep}{0pt}
  \begin{enumerate}[(1)]
  \item $R \in \FOV{2}[<]$
  \item $R$ is expressible in unary temporal logic
  \item $R \in \Sigma_2 \cap \Pi_2[<]$
  \item\label{UL}  $R$ is an unambiguous regular language
  \item The syntactic semi-group of $R$ is a member of {\bf DA}
  \item $R$ is recognizable by a partially-ordered 2-way automaton
  \item\label{turtle} $R$ is a boolean combination of ``turtle languages''
  \end{enumerate}
\end{fact}

The proofs of our structure theorems are self-contained applications
of \ef{} games.  All of the above characterizations follow from these
results.  Furthermore, we have now exactly connected quantifier and
alternation depth to the picture, thus adding tight bounds and further
insight to the above results.

For example, one can best understand item \ref{UL} above -- that
$\FOV{2}[<]$ on words corresponds to the unambiguous regular languages
-- via
Theorem \ref{thm:unique-rankers} which states that any $\FOV{2}_n[<]$
formula with one free variable that is always true of at most one position
in any string, necessarily denotes an $n$-ranker.

In the conclusion of \cite{STV01}, the authors define the subclasses
of rankers with one and two blocks of alternation.  They write that,
``\ldots turtle languages might turn out to be a helpful tool for
further studies in algebraic language theory.''  We feel that the
present paper fully justifies that prediction.  Turtle languages --- aka
rankers --- do provide an exceptionally clear and precise
understanding of the expressive power of $\fo^2$ on words, with and
without successor.

In summary, our structure theorems provide a complete classification of
the expressive power of $\FOV{2}$ on words in terms of both quantifier
depth and alternation.  They also tighten several previous
characterizations and lead to the alternation hierarchy results.

We begin the remainder of this paper with a brief review
of logical background including \ef{} games, our main tool. In 
Sect.~\ref{sec:FO2} we formally define rankers and present our structure
theorem for $\FOVD[$<$]{2}{n}$. The structure theorem for
$\FOVDA[$<$]{2}{n}{m}$ is covered in Sect.~\ref{sec:FO2-alt},
including our alternation hierarchy result that follows from
it. Sect.~\ref{sec:FO2-suc} extends our structure theorems and the
alternation hierarchy result to $\FOV[$<,\suc$]{2}$. Finally, we discuss
applications of our structural results to satisfiability for $\FOV[$<$]{2}$ in
Sect.~\ref{sec:satisfiability}.

\section{Background and Definitions} \label{sec:background}


We recall some notation concerning strings, first-order logic,
and \ef{} games.  See \cite{I99} for more details, including the
proof of Facts \ref{fact:ef1} and \ref{fact:ef2}.

$\Sigma$ will always denote a finite alphabet and $\varepsilon$ the empty
string. For a word $w \in \Sigma^\ell$ and $i \in [1, \ell]$, let $w_i$ be
the $i$-th letter of $w$; and for $[i,j]$ a subinterval of $[1, \ell]$, let
$w_{[i,j]}$ be the substring $w_i\ldots w_j$. Slightly abusing notation, we
identify a word $w \in \Sigma^\ell$ with the logical structure $w =
(\set{1,\ldots, \ell}; Q^w_{\mathtt{a}}, \mathtt{a} \in \Sigma; x^w; y^w)$.
Here $Q_\mathtt{a}, \mathtt{a} \in \Sigma$ are all unary relation symbols,
and $x$ and $y$ are the only two variables. If not specified otherwise, we
have $x^w = y^w = 1$ by default, and for all $\mathtt{a} \in \Sigma$,
$Q^w_{\mathtt{a}} = \{1 \le i \le \ell \mid w_i = \mathtt{a}\}$.
Furthermore, we write $(w, i, j)$ for the word structure $w$ with the two
variables set to $i$ and $j$, respectively, and $(w, i)$ for the word
structure $w$ with $x^w = i$. Thus $w = (w, 1, 1)$, and $(w, i) \models
Q_\mathtt{a}(x)$ iff $w_i = \mathtt{a}$.

We use $\FO[$<$]$ to denote first-order logic with a binary linear order
predicate $<$, and $\FO = \FO[$<, \suc$]$ for first-order logic with an
additional binary successor predicate. $\fo^2_n$ refers to the restriction
of first-order logic to use at most two distinct variables, and quantifier
depth $n$. $\fo^2_{m,n}$ is the further restriction to formulas such that
any path in their parse tree has at most $m$ blocks of alternating
quantifiers, and $\FOVA{2}{m} = \bigcup_{n \ge m} \FOVDA{2}{n}{m}$. We write
$u \equiv^2_n v$ to mean that $u$ and $v$ agree on all formulas from
$\FOVD{2}{n}$, and $u \equiv^2_{m,n} v$ if they agree on
$\FOVDA{2}{n}{m}$.

We assume that the reader is familiar with our main tool: the \ef{} game. In
each of the $n$ moves of the game $\fo^2_n(u,v)$, Samson places one of the
two pebble pairs, $x$ or $y$ on a position in one of the two words and
Delilah then answers by placing that pebble's mate on a position of the
other word. Samson wins if after any move, the map from the chosen points in
$u$ to those in $v$, i.e., $x^u\mapsto x^v$, $y^u\mapsto y^v$ is not an
isomorphism of the induced substructures; and Delilah wins otherwise. The
fundamental theorem of \ef{} games is the following:

\begin{fact}\label{fact:ef1}
  Let $u,v\in \Sigma^\star$, $n\in\N$.
  Delilah has a winning strategy for the game $\fo^2_n(u,v)$ iff
  $u \equiv^2_n v$.
\end{fact}

Thus, \ef{} games are a perfect tool for determining what is
expressible in first-order logic with a given quantifier-depth and
number of variables.  The game $\fo^2_{m,n}(u,v)$ is the restriction
of the game $\fo^2_n(u,v)$ in which Samson may change which word he
plays on at most $m-1$ times.

\begin{fact}\label{fact:ef2}
  Let $u,v\in \Sigma^\star$ and let $m, n \in \N$ with $m \le n$.
  Delilah has a winning strategy for the game $\fo^2_{m,n}(u,v)$ iff $u
  \equiv^2_{m,n} v$.
\end{fact}

We end this section with a simple lemma that will be useful whenever we want
to prove that there is a formula expressing a property of strings. With this
lemma, it suffices to show that for any pair of strings, one with the
property in question and one without, there is a formula that distinguishes
between these two particular strings.

\begin{lem} \label{lem:fin-equiv}
  Let $P \subseteq \Sigma^\star$ and let $L$ be a logic closed
  under boolean operations with only finitely many inequivalent formulas. If
  for every $u \in P$ and every $v \in \overline{P}$ there is a formula
  $\varphi_{u,v} \in L$ such that $u \models \varphi_{u,v}$ and $v
  \not\models \varphi_{u,v}$, then there is a formula $\varphi \in L$ such
  that for all $w \in \Sigma^\star$, $w \models \varphi \iff w \in P$.
\end{lem}

\begin{proof}
  Let $\Gamma = \{\psi_{u,v} \mid u \in P, v \in \overline{P} \}$, and
  let $\Gamma'$ be a maximal subset of $\Gamma$ containing only inequivalent
  formulas. Since $L$ contains only finitely many inequivalent formulas,
  $\Gamma'$ is finite. For every $u \in P$, we define the finite sets of
  formulas $\Gamma'_u = \{ \psi \in \Gamma' \mid u \models \psi\}$.
  Since all these sets are subsets of the finite set $\Gamma'$, there can
  only be finitely many of them. Thus there is a finite set $P' \subseteq P$
  such that $\{\Gamma'_u \mid u \in P\} = \{\Gamma'_u \mid u \in P'\}$. Now
  we set
  \begin{equation*}
    \varphi = \bigvee_{u \in P'} \bigwedge_{\psi \in \Gamma'_u} \psi
  \end{equation*}
  We have $\varphi \in L$ and for every $w \in \Sigma^\star$, $w \in P \iff
  w \models \varphi$ as required.
\end{proof}

It is well-known \cite{I99} that for any $m,n \in \N$, the logics
$\FOVD{2}{n}$ and $\FOVDA{2}{n}{m}$, both with and without the successor
predicate, have only finitely many inequivalent formulas. Thus the above
lemma applies to these logics.

\section{Structure Theorem for \texorpdfstring{\FOV[$<$]{2}}{FOV[<]2}} \label{sec:FO2}

We define boundary positions that point to the first or last occurrences of
a letter in a word, and define an $n$-ranker as a sequence of $n$ boundary
positions. In terms of \cite{STV01}, boundary positions are turtle
instructions and $n$-rankers are turtle programs of length $n$. The
following three lemmas show that basic properties about the definedness and
position of these rankers can be expressed in $\FOV[$<$]{2}$, and we use these
results to prove our structure theorem.

\begin{defi} \label{def:boundary-position}
  A \emph{boundary position} denotes the first or last
  occurrence of a letter in a given word. Boundary positions are of the form
  $d_\mathtt{a}$ where $d \in \{\rright, \rleft\}$ and $\mathtt{a} \in \Sigma$. The
  interpretation of a boundary position $d_\mathtt{a}$ on a word $w = w_1 \ldots
  w_{|w|} \in \Sigma^\star$ is defined as follows.
  \begin{displaymath}
    d_\mathtt{a}(w) =
    \begin{cases}
      \min\{i \in [1, |w|] \mid w_i = \mathtt{a}\} & \text{if } d = \rright\\
      \max\{i \in [1, |w|] \mid w_i = \mathtt{a}\} & \text{if } d = \rleft
    \end{cases}
  \end{displaymath}
  Here we set $\min\{\}$ and $\max\{\}$ to be undefined, thus $d_\mathtt{a}(w)$ is
  undefined if $a$ does not occur in $w$. A boundary position can also be
  specified with respect to a position $q \in [1, |w|]$.
  \begin{displaymath}
    d_\mathtt{a}(w, q) =
    \begin{cases}
      \min\{i \in [q+1, |w|] \mid w_i = \mathtt{a}\} & \text{if } d = \rright\\
      \max\{i \in [1, q-1] \mid w_i = \mathtt{a}\} & \text{if } d = \rleft
    \end{cases}
  \end{displaymath}
\end{defi}

\begin{defi} \label{def:n-ranker}
  Let $n$ be a positive integer. An \emph{$n$-ranker} $r$ is a sequence of
  $n$ boundary positions. The interpretation of an $n$-ranker $r = (p_1,
  \ldots, p_n)$ on a word $w$ is defined as follows.
  \begin{displaymath}
    r(w) :=
    \begin{cases}
      p_1(w) & \text{if } r = (p_1)\\
      \text{undefined} & \text{if $(p_1, \ldots, p_{n-1})(w)$ is
        undefined}\\
      p_n(w, (p_1, \ldots, p_{n-1})(w)) & \text{otherwise}
    \end{cases}
  \end{displaymath}
\end{defi}

Instead of writing $n$-rankers as a formal sequence $(p_1, \ldots, p_n)$, we
often use the simpler notation $p_1 \ldots p_n$. We denote the set of all
$n$-rankers by $R_n$, and the set of all $n$-rankers that
are defined over a word $w$ by $R_n(w)$. Furthermore, we set
$R_n^\star := \bigcup_{i \in [1, n]} R_i$ and $R_n^\star(w)
:= \bigcup_{i \in [1, n]} R_i(w)$.

\begin{defi}
  Let $r$ be an $n$-ranker. As defined above, we have $r = (p_1, \ldots,
  p_n)$ for boundary positions $p_i$. The \emph{$k$-prefix ranker} of $r$ for
  $k \in [1, n]$ is  $r_k := (p_1, \ldots, p_k)$.
\end{defi}

\begin{defi} \label{def:order-type}
  Let $i, j \in \N$. The \emph{order type} of $i$ and $j$ is defined as
  \begin{displaymath}
    \ord(i,j) =
    \begin{cases}
      < & \text{if } i < j\\
      = & \text{if } i = j\\
      > & \text{if } i > j\\
    \end{cases}
  \end{displaymath}
\end{defi}

\begin{lem}[distinguishing points on opposite sides of a ranker]
  \label{lem:ranker-sides}
  Let $n$ be a positive integer, let $u, v \in \Sigma^\star$ and let $r \in
  R_n(u) \cap R_n(v)$. Samson wins the game $\FOVD{2}{n}(u,v)$ where
  initially $\ord(x^u, r(u)) \ne \ord(x^v, r(v))$.
\end{lem}

\begin{proof}
  We only look at the case where $x^u \ge r(u)$ and $x^v < r(v)$ since all
  other cases are symmetric to this one. For $n = 1$ Samson has a winning
  strategy: If $r$ is the first occurrence of a letter, then Samson places
  $y$ on $r(u)$ and Delilah cannot reply. If $r$ marks the last occurrence
  of a letter in the whole word, then Samson places $y$ on $r(v)$. Again,
  Delilah cannot reply with any position and thus loses.

  \piccaption{\label{fig:ranker-sides-2}The case $r_{n-1}(u) < r(u)$}
  \parpic(7cm,3.3cm)[fr]{
      \begin{tikzpicture}
        \word[6cm]{\wordu}{$u$}
        \word[6cm]{\wordv}{$v$}
        \dwordpos{4cm}{$r$}
        \dwordpos{2cm}{$r_{n-1}$}
        \wordupoint{5cm}{$x$}
        \wordvpoint{3cm}{$x$}
        \wordupoint{4cm}{$S:y$}
        \wordvpoint{1cm}{$D:y$}
      \end{tikzpicture}
    }

  For $n > 1$, we look at the prefix ranker $r_{n-1}$ of $r$. One of the
  following two cases applies.
  \begin{enumerate}[(1)]
  \item
    \parpic(7cm,2.8cm)[r]{} $r_{n-1}(u) < r(u)$, as shown in
    Fig.~\ref{fig:ranker-sides-2}. Samson places pebble $y$ on $r(u)$, and
    Delilah has to reply with a position that is
    to the left of $x^v$. She cannot
    choose a position in the interval $(r_{n-1}(v), r(v))$, because this
    section does not contain the letter $u_{r(u)}$. Thus she has to choose a
    position left of or equal to $r_{n-1}(v)$. By induction Samson wins the
    remaining game.\medskip

  \item
    \piccaption{\label{fig:ranker-sides-3}The case $r(u) < r_{n-1}(u)$}
    \parpic(7cm,3.3cm)[fr]{
      \begin{tikzpicture}
        \word[6cm]{\wordu}{$u$}
        \word[6cm]{\wordv}{$v$}
        \dwordpos{2cm}{$r$}
        \dwordpos{4cm}{$r_{n-1}$}
        \wordupoint{3cm}{$x$}
        \wordvpoint{1cm}{$x$}
        \wordvpoint{2cm}{$S:y$}
        \wordupoint{5cm}{$D:y$}
      \end{tikzpicture}
    }
    $r(u) < r_{n-1}(u)$, as shown in Fig.
    \ref{fig:ranker-sides-3}. Samson places $y$ on $r(v)$, and Delilah has to
    reply with a position to the right of $x^u$ and thus to the right of $r(u)$.
    She cannot choose any position in
    $(r(u), r_{n-1}(u))$, because this interval does not
    contain the letter $v_{r(v)}$, thus Delilah has to choose a position
    to the right of or equal to $r_{n-1}(u)$. By induction Samson wins the
    remaining game.\qedhere
  \end{enumerate}
\end{proof}

\pagebreak[3]

\begin{lem}[expressing the definedness of a ranker] \label{lem:ranker-def}
  Let $n$ be a positive integer, and let $r \in R_n$. There is a
  formula $\varphi_r \in \FOVD[$<$]{2}{n}$ such that for all $w \in
  \Sigma^\star$, $w \models \varphi_r \iff r \in R_n(w)$.
\end{lem}

\begin{proof}
  Using Lemma $\ref{lem:fin-equiv}$ it suffices to consider arbitrary $u, v
  \in \Sigma^\star$ with $r \in R_n(u)$ and $r \notin R_n(v)$, and using
  Fact \ref{fact:ef1}, it suffices to show that Samson wins the game
  $\FOVD{2}{n}(u,v)$. If $r_1$, the shortest prefix ranker of $r$, is not
  defined over $v$, the letter referred to by $r_1$ occurs in $u$ but does
  not occur in $v$. Thus Samson easily wins in one move.

  \piccaption{\label{fig:rankpos-4}$r_i(v)$ is undefined}
  \parpic(5cm,3.3cm)[fr]{
    \begin{tikzpicture}
      \word[4cm]{\wordu}{$u$}
      \word[4cm]{\wordv}{$v$}
      \wordpos[$r_i$]{\wordu}{1cm}
      \dwordpos{2cm}{$r_{i-1}$}
      \wordupoint{1cm}{$S:x$}
      \wordvpoint{3cm}{$D:x$}
    \end{tikzpicture}
  }
  Otherwise we let $r_i = (p_1, \ldots, p_i)$ be the shortest prefix ranker of
  $r$ that is undefined over $v$. Thus $r_{i-1}$ is defined over both words.
  Without loss of generality we assume that $p_i = \rleft_\mathtt{a}$. This
  situation is illustrated in Fig.~\ref{fig:rankpos-4}. Notice that $v$
  does not contain any \texttt{a}'s to the left of $r_{i-1}(v)$, otherwise
  $r_i$ would be defined over $v$. Samson places $x$ in $u$ on $r_i(u)$, and
  Delilah has to reply with a position right of or equal to $r_{i-1}(v)$.
  Now Lemma \ref{lem:ranker-sides} applies and Samson wins in $i-1$ more
  moves.\qedconf
\end{proof}

\begin{lem}[position of a ranker]
  \label{lem:ranker-exp}
  Let $n$ be a positive integer and let $r \in R_n$. There is a formula
  $\psi_r \in \FOVD[$<$]{2}{n}$ such that for all $w \in \Sigma^\star$ and
  for all $i \in [1, |w|]$, $(w, i) \models \psi_r$ $\iff$ $i = r(w)$.
\end{lem}

\begin{proof}
  As in the proof of Lemma \ref{lem:ranker-def}, it suffices to show that 
  for arbitrary $u, v \in \Sigma^\star$, Samson wins the game
  $\FOVD{2}{n}(u, v)$ where initially $x^u = r(u)$
  and $x^v \ne r(v)$. If $r(v)$ is
  defined over $v$, then we can apply Lemma \ref{lem:ranker-sides}
  immediately to get the desired strategy for Samson. Otherwise we use the
  strategy from Lemma \ref{lem:ranker-def}.\qedconf
\end{proof}

\begin{thm}[structure of {\FOVD[$<$]{2}{n}}] \label{thm:ranker-char}
  Let $u$ and $v$ be finite words, and let $n \in \N$. The following two
  conditions are equivalent.
  \begin{enumerate}[\em(i)]
  \item
    \begin{enumerate}[\em(a)]
    \item $R_n(u) = R_n(v)$, and,
    \item for all $r \in R_n^\star(u)$ and $r'
      \in R_{n-1}^\star(u)$, $\ord(r(u), r'(u)) = \ord(r(v), r'(v))$
    \end{enumerate}
  \item $u \equiv^2_n v$
  \end{enumerate}
\end{thm}

Notice that condition (i)(a) is equivalent to $R_{n}^\star(u) =
R_{n}^\star(v)$. 
Instead of proving Theorem \ref{thm:ranker-char} directly, we prove the
following more general version on words with two interpreted variables.

\begin{thm} \label{thm:structure-variables}
  Let $u$ and $v$ be finite words, let $i_1, i_2 \in [1,|u|]$, let $j_1,
  j_2 \in [1,|v|]$, and let $n \in \N$. The following two conditions are
  equivalent.
  \begin{enumerate}[\em(i)]
  \item
    \begin{enumerate}[\em(a)]
    \item $R_n(u) = R_n(v)$, and,
    \item for all $r \in R_n^\star(u)$ and $r'
      \in R_{n-1}^\star(u)$, $\ord(r(u), r'(u)) = \ord(r(v), r'(v))$, and,
    \item $(u, i_1, i_2) \equiv^2_0 (v, j_1, j_2)$, and,
    \item for all $r \in R_n^\star(u)$, $\ord(i_1, r(u)) = \ord(j_1,
      r(v))$ and $\ord(i_2, r(u)) = \ord(j_2, r(v))$
    \end{enumerate}
  \item $(u, i_1, i_2) \equiv^2_n (v, j_1, j_2)$
  \end{enumerate}
\end{thm}

\begin{proof}
  For $n=0$, (i)(a), (i)(b) and (i)(d) are vacuous, and (i)(c) is equivalent
  to (ii). For $n \ge 1$, we prove the two implications individually using
  induction on $n$.

  We first show ``$\neg \text{(i)} \Rightarrow \neg \text{(ii)}$''. Assuming
  that (i) holds for $n \in \N$ but fails for $n+1$, we show that $(u,
  i_1, i_2) \not\equiv^2_{n+1} (v, j_1, j_2)$ by giving a winning
  strategy for Samson in the \FOVD{2}{n+1} game on the two structures. If (i)(c)
  does not hold, then Samson wins immediately. If (i)(d) does not
  hold for $n+1$, then Samson wins by Lemma \ref{lem:ranker-sides}. If
  (i)(a) or (i)(b) do not hold for $n+1$, then one of the following three
  cases applies.
  \begin{enumerate}[(1)]
  \item There is an $(n+1)$-ranker that is defined over one word but not
    over the other.
  \item There are two $n$-rankers that do not agree on their ordering in $u$
    and $v$.
  \item There is an $(n+1)$-ranker that does not appear in the same order on
    both structures with respect to a $k$-ranker where $k \le n$.
  \end{enumerate}

  We first look at case (2) where there are two rankers $r, r' \in
  R_n^\star(u)$ that disagree on their ordering in $u$ and $v$.
  Without loss of generality
  we assume that $r(u) \le r'(u)$ and $r(v) > r'(v)$, and present a winning
  strategy for Samson in the $\FOVD{2}{n+1}$ game. In the first move he places
  $x$ on $r(u)$ in $u$. Delilah has to reply with $r(v)$ in $v$, otherwise
  she would lose the remaining $n$-move game as shown in Lemma
  \ref{lem:ranker-sides}. Let $r'_{n-1}$ be the $(n-1)$-prefix-ranker of
  $r'$. We look at two different cases depending on the ordering of
  $r'_{n-1}$ and $r'$.

  \piccaption{\label{fig:char-order-1}Two $n$-rankers appear in different
    order and $r'$ ends with $\rright$.}
  \parpic(7cm,3.3cm)[fr]{
    \begin{tikzpicture}
      \word[6cm]{\wordu}{$u$}
      \word[6cm]{\wordv}{$v$}
      \dwordpos{2cm}{$r'_{n-1}$}
      \wordupos{3cm}{$r$}
      \dwordpos{4cm}{$r'$}
      \wordvpos{5cm}{$r$}
      \wordupoint{3cm}{$S:x$}
      \wordvpoint{5cm}{$D:x$}
      \wordvpoint{4cm}{$S:y$}
      \wordupoint{1cm}{$D:y$}
    \end{tikzpicture}
  } 
  
  For $r'_{n-1}(u) < r'(u)$, the situation is illustrated in Fig.
  \ref{fig:char-order-1}. In his second move, Samson places $y$ on $r'(v)$.
  Delilah has to reply with a position to the left of $x^u$, but she cannot
  choose any position from the interval $(r'_{n-1}(u), r'(u))$ because it
  does not contain the letter $v_{y^v}$. So she has to reply with
  a position left of or equal to $r'_{n-1}(u)$, and Samson wins the
  remaining \FOVD{2}{n-1} game as shown in Lemma \ref{lem:ranker-sides}.\bigskip

  \piccaption{\label{fig:char-order-2}Two $n$-rankers appear in different
    order and $r'$ ends with $\rleft$.}
  \parpic(7cm,3.3cm)[fr]{
    \begin{tikzpicture}
      \word[6cm]{\wordu}{$u$}
      \word[6cm]{\wordv}{$v$}
      \dwordpos{4cm}{$r'_{n-1}$}
      \wordupos{1cm}{$r$}
      \dwordpos{2cm}{$r'$}
      \wordvpos{3cm}{$r$}
      \wordupoint{1cm}{$S:x$}
      \wordvpoint{3cm}{$D:x$}
      \wordupoint{2cm}{$S:y$}
      \wordvpoint{5cm}{$D:y$}
    \end{tikzpicture}
  }

  For $r'_{n-1}(u) > r'(u)$, the situation is illustrated in Fig.
  \ref{fig:char-order-2}. In his second move, Samson places pebble $y$ on $r'(u)$, and
  Delilah has to reply with a position to the right of $x^v$, but she cannot
  choose anything from the interval $(r'(v), r'_{n-1}(v))$ because this
  section does not contain the letter $u_{y^u}$. Thus she has to reply with
  a position right of or equal to $r'_{n-1}(v)$, and Samson wins the
  remaining \FOVD{2}{n-1} game as shown in Lemma \ref{lem:ranker-sides}.
  
  \piccaption{\label{fig:char-final}A letter \texttt{a} occurs between
    $n$-rankers $r, r'$ in $u$ but not in $v$}
  \parpic(5.8cm,3.3cm)[fr]{
    \begin{tikzpicture}
      \word[4cm]{\wordu}{$u$}
      \word[4cm]{\wordv}{$v$}
      \dwordpos{1cm}{$r$}
      \dwordpos{3cm}{$r'$}
      \wordupoint{2cm}{$S:x$}
      \worduletter{2cm}{\texttt{a}}
    \end{tikzpicture}
  }

  \picskip{11}
  Now we look at cases (1) and (3), assuming that case (2) does not apply.
  We know that condition (i) from the statement of the theorem fails, but
  still all $n$-rankers agree on their ordering. In both case (1) and case (3),
  there are two
  consecutive $n$-rankers $r, r' \in R_{n}(u)$ with $r(u) < r'(u)$ and a
  letter $\mathtt{a} \in \Sigma$ such that without loss of generality
  \texttt{a} occurs in the segment $u_{((r(u), r'(u))}$ but not in the
  segment $v_{(r(v), r'(v))}$. We describe a winning strategy for Samson in
  the game $\FOVD{2}{n+1}(u,v)$. He places $x$ on an \texttt{a} in the
  segment $(r(u), r'(u))$ of $u$, as shown in Fig.~\ref{fig:char-final}.
  Delilah cannot reply with anything in the interval $(r(v), r'(v))$. If she
  replies with a position left of or equal to $r(v)$, then $x$ is on
  different sides of the $n$-ranker $r$ in the two words. Thus Lemma
  \ref{lem:ranker-sides} applies and Samson wins the remaining $n$-move
  game. If Delilah replies with a position right of or equal to $r'(v)$,
  then we can apply Lemma \ref{lem:ranker-sides} to $r'$ and get a winning
  strategy for the remaining game as well. This concludes the proof of
  ``$\neg \text{(i)} \Rightarrow \neg \text{(ii)}$''.

  To show ``(i) $\Rightarrow$ (ii)'', we assume (i) for $n+1$, and present a
  winning strategy for Delilah in the $\FOVD{2}{n+1}$ game on the two
  structures. In his first move Samson picks up one of the two pebbles, and
  places it on a new position. Without loss of generality we assume that
  Samson picks up $x$ and places it on $u$ in his first move. If $x^u =
  r(u)$ for any ranker $r \in R_{n+1}^\star(u)$, then Delilah replies with
  $x^v = r(v)$. This establishes (i)(c) and (i)(d) for $n$, and thus
  Delilah has a winning strategy for the remaining \FOVD{2}{n} game by
  induction.

  If Samson does not place $x^u$ on any ranker from $R_{n+1}^\star(u)$,
  then we look at the closest rankers from $R_n^\star(u)$ to the left and
  right of $x^u$, denoted by $\lambda$ and $\rho$, respectively. Let
  $\mathtt{a} := u_{x^u}$ and define the $(n+1)$-ranker $s = (\lambda,
  \rright_\mathtt{a})$. On $u$ we have $\lambda(u) < s(u) < \rho(u)$. Because
  of (i)(a) $s$ is defined on $v$ as well, and because of (i)(b), we have
  $\lambda(v) < s(v) < \rho(v)$. If $y^u$ is not contained in the interval
  $(\lambda(u), \rho(u))$, then Delilah places $x$ on $s(v)$, which
  establishes (i)(c) and (i)(d) for $n$. Thus by induction Delilah has a
  winning strategy for the remaining \FOVD{2}{n} game.

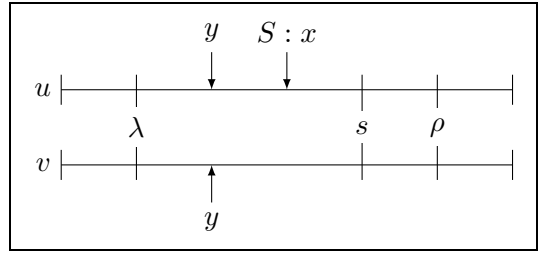
  \piccaption{\label{fig:rankchar-1}$x$ and $y$ are in the same section}
  \parpic(7cm,3.3cm)[fr]{
    \begin{tikzpicture}
      \word[6cm]{\wordu}{$u$}
      \word[6cm]{\wordv}{$v$}
      \dwordpos{1cm}{$\lambda$}
      \dwordpos{4cm}{$s$}
      \dwordpos{5cm}{$\rho$}
      \wordupoint{2cm}{$y$}
      \wordvpoint{2cm}{$y$}
      \wordupoint{3cm}{$S:x$}
    \end{tikzpicture}
  }

  \picskip{10}
  If both pebbles $x^u$ and $y^u$ occur in the interval $(\lambda(u),
  \rho(u))$, then we need to be more careful. Without loss of generality we
  assume $y^u < x^u$ as illustrated in Fig.~\ref{fig:rankchar-1}. Thus
  Delilah has to place $x$ in the interval $(y^v, \rho(v))$ and at
  a position with letter $\mathtt{a} := u_{x^u}$. We define the
  $n+1$-ranker $s = (\rho, \rleft_\mathtt{a})$. From (i)(d) we know that $s$
  appears on the same side of $y$ in both structures, thus we have $y^v <
  s(v) < \rho(v)$. Delilah places her pebble $x$ on $s(v)$, and thus
  establishes (i)(c) and (i)(d) for $n$. By induction, Delilah has a winning
  strategy for the remaining \FOVD{2}{n} game.\qedconf
\end{proof}

A fundamental property of an $n$-ranker is that it uniquely describes a
position in a given word. Now we show that the converse holds as well: Any
position in a word that can be uniquely described with an \FOV[$<$]{2}
formula can also be described by a ranker (Lemma \ref{lem:unique-rankers}).
Furthermore, any \FOV[$<$]{2} formula that describes a unique position in
any given word is equivalent to a boolean combination of rankers (Theorem
\ref{thm:unique-rankers}).

\begin{defi}[unique position formula] \label{def:unique-pos}
  A formula $\varphi \in \FOV[$<$]{2}$ with $x$ as a free variable is a
  \emph{unique position formula} if for all $w \in \Sigma^\star$ there is at
  most one $i \in [1, |w|]$ such that $(w, i) \models \varphi$.
\end{defi}

\begin{lem} \label{lem:unique-rankers} 
  Let $n$ be a positive integer and let $\varphi \in \FOVD[$<$]{2}{n}$ be a
  unique position formula. Let $u \in \Sigma^\star$ and let $i \in [1, |u|]$
  such that $(u, i) \models \varphi$. Then $i = r(u)$ for some ranker $r \in
  R_n^\star$.
\end{lem}

\begin{full}
\begin{proof}
  Suppose for the sake of a contradiction that there is no ranker $r \in
  R_n^\star$ such that $(u, i) \models \varphi_r$. Because the first and
  last positions in $u$ are described by 1-rankers, we know that $i \notin
  \{1, |u|\}$.
  We construct a new word $v$ by doubling the symbol at position $i$ in $u$,
  $v = u_1 \ldots u_{i-1} u_i u_i u_{i+1} \ldots u_{|u|}$. By assumption,
  there is no $n$-ranker that describes position $i$ in $u$. A brief
  argument by contradiction shows that there are also no $n$-rankers that
  describe positions $i$ or $i+1$ in $v$: Assuming that such a ranker
  exists, let $r$ be the shortest such ranker. Thus none of the prefix
  rankers of $r$ point to either positions $i$ or $i+1$ in $v$. This means
  that all prefix rankers of $r$ are interpreted in exactly the same way on
  both $u$ and $v$, and irrespective of whether $r(v)$ points to $i$ or
  $i+1$, we have have $r(u) = i$, a
  contradiction. Hence all $n$-rankers are insensitive to the doubling of
  $u_i$, and the two
  words $u$ and $v$ agree on the definedness of all $n$-rankers and on their
  ordering. 
  By
  Theorem \ref{thm:structure-variables}, we thus have $(u, i) \equiv^2_n
  (v, i) \equiv^2_n (v, i+1)$, which contradicts the fact that $\varphi$
  is a unique position formula.\qedconf
\end{proof}
\end{full}

\begin{thm} \label{thm:unique-rankers}
  Let $n$ be a positive integer and let $\varphi \in \FOVD[$<$]{2}{n}$ be a
  unique position formula. There is a $k \in \N$, and there
  are mutually exclusive formulas $\alpha_i \in \FOVD[$<$]{2}{n}$ and rankers
  $r_i \in R_n^\star$ such that
  \begin{displaymath}
    \varphi \equiv 
    \bigvee_{i \in [1,k]} \left( \alpha_i \land \varphi_{r_i} \right)
  \end{displaymath}
  where $\varphi_{r_i} \in \FOVD[$<$]{2}{n}$ is the formula from Lemma
  \ref{lem:ranker-exp} that uniquely describes the ranker $r_i$.
\end{thm}

\begin{full}
\begin{proof}
  Let $\mathcal{T}$ be the set of all \FOVD[$<$]{2}{n} types of words over
  $\Sigma$ with one interpreted variable. Because there are only finitely
  many inequivalent formulas in \FOVD[$<$]{2}{n}, $\mathcal{T}$ is finite. Let
  $\mathcal{T}' \subseteq \mathcal{T}$ be the set of all types that satisfy
  $\varphi$. We set $\mathcal{T'} = \{T_1, \ldots, T_k\}$ and let $\alpha_i
  \in \FOVD[$<$]{2}{n}$ be a description of type $T_i$.  Thus $\varphi
  \equiv \bigvee_{i \in [1,k]} \alpha_i$.

  Now suppose that $(u, j) \models \varphi$. Thus $(u, j) \models \alpha_i$
  for some $i$. By Lemma \ref{lem:unique-rankers} $(u, j) \models
  \varphi_{r_i}$ for some $r_i\in R_n^\star$. Thus $\alpha_i \rightarrow
  \varphi_{r_i}$ since $\varphi_{r_i}\in \fo^2_n$ and $\alpha_i$ is a complete
  $\fo^2_n$ formula. Thus $\alpha_i \equiv \alpha_i \land \varphi_{r_i}$ so
  $\varphi$ is in the desired form.\qedconf
\end{proof}
\end{full}

\section{Alternation hierarchy for \texorpdfstring{\FOV[$<$]{2}}{FOV[<]2}} \label{sec:FO2-alt}

We define alternation rankers and prove our structure theorem (Theorem
\ref{thm:alt-ranker-char}) for $\FOVDA[$<$]{2}{n}{m}$. Surprisingly the number
of alternating blocks of $\rleft$ and $\rright$ in the rankers corresponds
exactly to the number of alternating quantifier blocks. The main ideas from
our proof of Theorem \ref{thm:ranker-char} still apply here, but keeping
track of the number of alternations does add complications.

\begin{defi}[$m$-alternation $n$-ranker] \label{def:alt-ranker} Let
  $m,n \in \N$ with $m \le n$. An \emph{$m$-alternation $n$-ranker}, or
  $(m,n)$-ranker, is an $n$-ranker with exactly $m$ blocks of boundary
  positions that alternate between $\rright$ and $\rleft$.
\end{defi}

We use the following notation for alternation rankers.
\begin{align*}
  R_{m,n}(w) &:= \{r \mid \text{$r$ is an $m$-alternation $n$-ranker and
    defined over the word $w$}\}\\
  R_{m\rright,n}(w) 
  &:= \{r \in R_{m,n}(w) \mid \text{$r$ ends with $\rright$}\}\\
  R^\star_{m,n}(w)
  &:= \bigcup_{i \in [1,m], j \in [1,n]} R_{i,j}(w)\\
  R^\star_{m\rright,n}(w)
  &:= R^\star_{m-1,n}(w) \cup \bigcup_{i \in [1,n]} R_{m\rright,i}(w)
\end{align*}

\begin{full}
\begin{lem} \label{lem:alt-ranker-sides} Let $m$ and $n$ be positive
  integers with $m \le n$, let $u, v \in \Sigma^\star$, and let $r \in
  R_{m,n}(u) \cap R_{m,n}(v)$. Samson wins the game $\FOVDA{2}{n}{m}(u,v)$
  where initially $\ord(r(u), x^u) \ne \ord(r(v), x^v)$.

  Furthermore, Samson can start the game with a move on $u$ if $r$ ends with
  $\rright$, $r(u) \le x^u$ and $r(v) \ge x^v$, or if $r$ ends with
  $\rleft$, $r(u) \ge x^u$ and $r(v) \le x^v$. He can start the game with
  a move on $v$ if $r$ ends with $\rright$, $r(u) \ge x^u$ and $r(v) \le
  x^v$, or if $r$ ends with $\rleft$, $r(u) \le x^u$ and $r(v) \ge x^v$.
\end{lem}

\begin{proof}
  If $m = n = 1$, then we can immediately apply the base case from the proof
  of Lemma \ref{lem:ranker-sides}. Samson wins in one move, placing his
  pebble on $u$ or $v$ as specified.

  For the remaining cases, we assume without loss of generality that $r$
  ends with $\rright$ and that $x^u \ge r(u)$ and $x^v \le r(v)$. Let
  $r_{n-1}$ be the $(n-1)$-prefix ranker of $r$. This situation is
  illustrated in Fig.~\ref{fig:ranker-sides-2} of Lemma
  \ref{lem:ranker-sides}. Samson places $y$ on $r(u)$, and creates a
  situation where $y^u > r_{n-1}(u)$ and $y^v \le r_{n-1}(v)$. If
  $r_{n-1}$ ends with $\rleft$, then by induction Samson wins the remaining
  \FOVDA{2}{n-1}{m-1} game and thus he has a winning strategy for the
  \FOVDA{2}{n}{m} game. If $r_{n-1}$ ends with $\rright$, then by induction
  Samson wins the remaining \FOVDA{2}{n-1}{m} game starting with a move on
  $u$, and thus he has a winning strategy for the \FOVDA{2}{n}{m} game.\qedconf
\end{proof}

\begin{lem} \label{lem:alt-ranker-def} Let $m$ and $n$ be positive
  integers with $m \le n$ and let $r \in R_{m,n}$. There is a 
  $\varphi_r \in \FOVDA[$<$]{2}{n}{m}$ such that for all $w \in \Sigma^\star$,
  $w \models \varphi_r \iff r \in R_{m,n}(w)$.
\end{lem}

\begin{proof}
  Using Lemma $\ref{lem:fin-equiv}$ it suffices to consider arbitrary $u, v
  \in \Sigma^\star$ with $r \in R_{m,n}(u)$ and $r \notin R_{m,n}(v)$, and
  using Fact \ref{fact:ef1}, it suffices to show that Samson wins the game
  $\FOVDA{2}{n}{m}(u,v)$. 
  Let $r_i = (p_1, \ldots, p_i)$ be the shortest prefix ranker
  of $r$ that is undefined over $v$, and we assume without loss of
  generality that this ranker ends with the boundary position $p_i =
  \rleft_\mathtt{a}$ for some $\mathtt{a} \in \Sigma$. This situation is illustrated in Fig.
  \ref{fig:rankpos-4} for Lemma \ref{lem:ranker-exp}. In his first move
  Samson places $x$ on $r_i(u)$ and thus forces a situation where $x^u <
  r_{i-1}(u)$ and $x^v \ge r_{i-1}(v)$. If $r_{i-1}$ ends with $\rleft$,
  then according to Lemma \ref{lem:alt-ranker-sides}, Samson wins the
  remaining \FOVDA{2}{n-1}{m} game starting with a move on $u$. Otherwise
  $r_{i-1}$ ends with $\rright$, and thus by Lemma
  \ref{lem:alt-ranker-sides} Samson wins the remaining
  \FOVDA{2}{n-1}{m-1} game
  starting with a move on $v$.\qedconf
\end{proof}

\begin{lem} \label{lem:alt-ranker-exp} Let $m$ and $n$ be positive
  integers with $m \le n$ and let $r \in R_{m,n}$. There is a formula
  $\psi_r \in \FOVDA[$<$]{2}{n}{m}$ such that for all $w \in \Sigma^\star$
  and for all $i \in [1, |w|]$, $(w, i) \models \psi_r \iff i = r(w)$.
\end{lem}

\begin{proof}
  As in the proof of Lemma \ref{lem:alt-ranker-def}, it suffices to show
  that Samson wins the game
  $\FOVDA{2}{n}{m}(u,v)$ where initially 
  $x^u=r(u)$ and $x^v \ne r(v)$. Depending on
  whether $r$ is defined over $v$, we use the strategies from Lemma
  \ref{lem:alt-ranker-sides} or Lemma \ref{lem:alt-ranker-def}.\qedconf
\end{proof}
\end{full}

\begin{thm}[structure of {\FOVDA[$<$]{2}{n}{m}}]
  \label{thm:alt-ranker-char}
  Let $u$ and $v$ be finite words, and let $m, n \in \N$ with $m \le n$. The
  following two conditions are equivalent.
  \begin{enumerate}[\em(i)]
  \item
    \begin{enumerate}[\em(a)]
    \item $R_{m,n}(u) = R_{m,n}(v)$, and,
    \item for all $r \in R_{m,n}^\star(u)$ and for all $r' \in
      R_{m-1,n-1}^\star(u)$, we have\\
      $\ord(r(u),r'(u)) = \ord(r(v), r'(v))$, and,
    \item for all $r \in R_{m,n}^\star(u)$ and $r' \in R_{m,n-1}^\star(u)$
      such that $r$ and $r'$ end with different directions,
      $\ord(r(u),r'(u)) = \ord(r(v), r'(v))$
    \end{enumerate}
  \item $u \equiv^2_{m,n} v$
  \end{enumerate}
\end{thm}

Just as before with Theorem \ref{thm:ranker-char}, instead of proving
Theorem \ref{thm:alt-ranker-char} directly, we prove a more general version
that applies to words with two interpreted variables. The statement of the
general version is asymmetric with respect to the roles of the two
structures $u$ and $v$. This is necessary because of the correspondence
between quantifier alternations (i.e. alternations between $u$ and $v$ in
the game) and alternations of directions in the rankers. This asymmetry
already affected the statement of Lemma \ref{lem:alt-ranker-sides}, where
Samson's winning strategy starts with a move on the specified structure. In
fact, as the proof of the following theorem shows, he does not have a
winning strategy that starts with a move on the other structure. We remark
that conditions (i)(a) through (i)(e) of the general theorem are completely
symmetric with respect to the roles of $u$ and $v$, and only conditions
(i)(f) and (ii) are asymmetric. Theorem \ref{thm:alt-ranker-char} follows
directly from the general theorem, since here $i_1 = i_2 = j_1 = j_2 = 1$,
thus conditions (i)(e) and (i)(f) or trivially true, and the equivalence
holds with the roles of $u$ and $v$ reversed as well.
\pagebreak[4]

\begin{thm}
  Let $u$ and $v$ be finite words, let $i_1, i_2 \in [1,|u|]$, let $j_1, j_2
  \in [1,|v|]$, and let $m, n \in \N$ with $m \le n$. The following two
  conditions are equivalent.

  \begin{enumerate}[\em(i)]
  \item
    \begin{enumerate}[\em(a)]
    \item $R_{m,n}(u) = R_{m,n}(v)$, and,
    \item for all $r \in R_{m,n}^\star(u)$ and for all $r' \in
      R_{m-1,n-1}^\star(u)$, we have\\
      $\ord(r(u),r'(u)) = \ord(r(v), r'(v))$, and,
    \item for all $r \in R_{m,n}^\star(u)$ and $r' \in R_{m,n-1}^\star(u)$
      such that $r$ and $r'$ end with different directions,
      $\ord(r(u),r'(u)) = \ord(r(v), r'(v))$
    \item $(u, i_1, i_2) \equiv^2_0 (v, j_1, j_2)$, and,
    \item for all $r \in R_{m-1,n}^\star(u)$, $\ord(r(u), i_1) = \ord(r(v),
      j_1)$ and $\ord(r(u), i_2) = \ord(r(v), j_2)$, and,
    \item for all $r \in R_{m,n}^\star(u)$, and $(i,j) \in \{(i_1, j_1),
      (i_2, j_2)\}$,
      \begin{enumerate}[\em(f${}_1$)]
      \item if $r$ ends on $\rright$ and $r(u) = i$, then $r(v) \le j$
      \item if $r$ ends on $\rright$ and $r(u) < i$, then $r(v) < j$
      \item if $r$ ends on $\rleft$ and $r(u) = i$, then $r(v) \ge j$
      \item if $r$ ends on $\rleft$ and $r(u) > i$, then $r(v) > j$
      \end{enumerate}
    \end{enumerate}
  \item Delilah wins the game $\FOVDA[$<$]{2}{n}{m}((u, i_1, i_2), (v, j_1,
    j_2))$ if Samson starts with a move on $(u, i_1, i_2)$.
  \end{enumerate}
\end{thm}

\begin{full}
\begin{proof}

  As in the proof of Theorem \ref{thm:ranker-char}, we use induction on $n$.
  For $n = 0$, condition (i)(d) just by itself is equivalent to (ii), and
  all other conditions of (i) are vacuous. For $n \ge 1$, we we first show
  ``$\neg$ (i) $\Rightarrow$ $\neg$ (ii)''.

  Suppose that (i) holds for $(m,n)$, but fails for $(m,n+1)$. If (i)(d)
  does not hold then Samson wins immediately. If (i)(e) does not hold for
  $(m,n+1)$, then by Lemma \ref{lem:alt-ranker-sides}, Samson wins the
  $(m,n+1)$-game on $(u,v)$, starting with a move on either $u$ or $v$. If
  Samson can start with a move on $u$, we have established that (ii) is
  false. Otherwise, we reverse the roles of $u$ and $v$, and observe that
  condition (i)(e) still remains the same. Thus, even if Samson needs to
  start with a move on $v$, he still has a winning strategy, and (ii) does
  not hold for $(m,n+1)$. If (i)(f) does not hold for $(m,n+1)$, then again
  by using Lemma \ref{lem:alt-ranker-sides}, Samson wins the $(m,n+1)$-game
  on $(u,v)$ starting with a move on $u$.

  If one of (i)(a), (i)(b) or (i)(c) fail, then we show that Samson has a
  winning strategy for
  the game $\FOVDA{2}{n+1}{m}(u,v)$. We observe that it does not matter what
  structure Samson chooses for his first move, since all of (i)(a), (i)(b)
  and (i)(c) are completely symmetric with respect to the roles of $u$ and
  $v$. Thus if Samson's winning strategy starts with a move on $v$, we can
  reverse the roles of $u$ and $v$ and get a winning strategy starting with
  move on $u$. One of the following cases applies.
  \begin{enumerate}[(1)]
  \item There is a ranker $r \in R_{m,n+1}$ that is defined over one
    structure but not over the other.
  \end{enumerate}
  This first case applies if (a) fails for $(m,n+1)$. If condition (2) fails
  for $(m, n+1)$, then there are two $n$-rankers for which it fails, or an
  $(n+1)$-ranker and an $n$-ranker. This leads to the following two cases.
  \begin{enumerate}[(1)]
    \setcounter{enumi}{1}
  \item There are two rankers $r \in R_{m,n}(u)$ and $r' \in R_{m-1,n}(u)$
    that disagree on their order, i.e. $\ord(r(u),r'(u)) \ne \ord(r(v),r'(v))$. 
  \item There are two rankers $r \in R_{m,n+1}(u)$ and $r' \in R_{m-1,n}(u)$
    that disagree on their order.
  \end{enumerate}
  In a similar fashion, we obtain the remaining two cases if condition (3)
  fails for $(m, n+1)$.
  \begin{enumerate}[(1)]
    \setcounter{enumi}{3}
  \item There are rankers $r,r' \in R_{m,n}(u)$ that end on different
    directions and disagree on their order.
  \item There are rankers $r \in R_{m,n+1}(u)$ and $r' \in R_{m,n}(u)$ that end
    on different directions and disagree on their order.
  \end{enumerate}

  We look at the
  cases (2) and (4) first, then deal with case (1) assuming that cases (2)
  and (4) do not apply, and finally look at cases (3) and (5).

  For case (2), we assume that $r(u) \le r'(u)$, as
  illustrated in Fig.~\ref{fig:alt-rankchar-order}. The situation for
  $r(u) \ge r'(u)$ is completely symmetric. Depending on the last boundary
  position of $r$, one of the following two subcases applies.

  \begin{enumerate}[$\bullet$]
    \piccaption{\label{fig:alt-rankchar-order}$r$ and $r'$ appear in different
      order}
    \parpic(5cm,2.7cm)[fr]{
      \begin{tikzpicture}
        \word[4cm]{\wordu}{$u$}
        \word[4cm]{\wordv}{$v$}
        \wordupos{1cm}{$r$}
        \dwordpos{2cm}{$r'$}
        \wordvpos{3cm}{$r$}
      \end{tikzpicture}
    }
  \item
    $r$ ends with $\rright$. Samson places $x$ on $r(u)$ in his first
    move. If Delilah replies with a position to the left of $r'(v)$ 
    or equal to $r'(v)$, then $x^v < r(v)$. Thus we
    can apply Lemma \ref{lem:alt-ranker-sides} to get a winning strategy for
    Samson in the remaining \FOVDA{2}{n}{m} game that starts with a move on $u$.
    If Delilah replies with a position to the right of $r'(v)$, Samson has a
    winning strategy for the remaining \FOVDA{2}{n}{m-1} game. Thus we have a
    winning strategy for Samson in the \FOVDA{2}{n+1}{m} game.
  \item
    \picskip{0}
    $r$ ends with $\rleft$. This is similar to the previous case, but
    now Samson places $x$ on $r(v)$ in his first move. If Delilah replies
    with a position to the right of $r'(u)$, or equal to $r'(u)$,
    then as above we get a winning
    strategy for Samson in the remaining \FOVDA{2}{n}{m} game that starts with a
    move on $v$. Otherwise we get a winning strategy for Samson with only
    $m-1$ alternations for the remaining game. Thus again he has a winning
    strategy for the \FOVDA{2}{n+1}{m} game.
  \end{enumerate}

  \picskip{0}
  For case (4), Samson's winning strategy is very similar to the previous
  case. If $r(u) \le r'(u)$ and $r$ ends with $\rright$, then Samson places
  $x$ on $r(u)$ in his first move. If Delilah replies with a position to the
  right of $r(u)$, then Samson's winning strategy is as above. Otherwise $x$
  is on different sides of $r'$ and Samson has a winning strategy for the
  remaining \FOVDA{2}{n}{m} game 
  that starts with a move on $u$. All together, he
  has a winning strategy for the \FOVDA{2}{n+1}{m} game.
  The remaining three cases (ordering of $r(u)$ and $r'(u)$ and ending
  direction of $r$) work in the same way.

  Similar to what we did in the proof of Theorem \ref{thm:ranker-char}, we
  can reduce the remaining cases to an easier situation where a certain
  segment contains a certain letter in one structure, but not in the other
  structure, and then apply Lemma \ref{lem:alt-ranker-sides} to obtain a
  winning strategy for Samson.\pagebreak[2]

  To deal with case (1), we assume that the previous two cases, (2) and (4),
  do not apply. Without loss of generality, say that the $(m,n+1)$-ranker $r$ is
  defined over $u$ but not over $v$. Let $\mathtt{a} := u_{r(u)}$ be the
  letter in $u$ at position $r(u)$. We define the following sets of rankers.
  \begin{align*}
    R_\ell &:= \{s \in R^\star_{m\rright, n}(u) \mid s(u) < r(u)\}\\
    R_r &:= \{s \in R^\star_{m\rleft, n}(u) \mid s(u) > r(u)\}
  \end{align*}
  Notice that all rankers from $R_\ell$ appear to the left of all rankers
  from $R_r$ in $u$. From the inductive hypothesis, and from the fact that
  both cases (2) and (4) do not apply, it follows that over $v$, all rankers
  from $R_\ell$ appear to the left of all rankers from $R_r$ as well.
  However, the rankers from $R_\ell$ and $R_r$ by themselves do
  not necessarily appear in the same order in both structures. We look at
  the ordering of these rankers in $v$, and let $\lambda$ be the rightmost
  ranker from $R_\ell$ and $\rho$ be the leftmost ranker from $R_r$.
  By construction, we have $\lambda(u) < r(u) < \rho(u)$, so
  the segment $(\lambda, \rho)$ in $u$ contains the letter $\mathtt{a}$. Let
  $r_n$ be the $n$-prefix-ranker of $r$, and observe that $r_n$ is defined
  on both structures and that $r_n$ is contained in either $R_\ell$ or
  $R_r$. Because $r$ is not defined on $v$, the letter $\mathtt{a}$ does not
  occur in $v$ either to the right of $r_n$ if $r_n \in R_\ell$, 
  or to the left of
  $r_n$ if $r_n \in R_r$. Thus the segment $(\lambda, \rho)$ does not contain the
  letter $\mathtt{a}$ in $v$.

  \piccaption{\label{fig:alt-ranker-char-middle}A letter occurs between
    rankers $r$, $r'$ in $u$ but not in $v$}
  \parpic(5.7cm,3.3cm)[fr]{
    \begin{tikzpicture}
      \word[4cm]{\wordu}{$u$}
      \word[4cm]{\wordv}{$v$}
      \dwordpos{1cm}{$\lambda$}
      \dwordpos{3cm}{$\rho$}
      \wordupoint{2cm}{$S:x$}
      \worduletter{2cm}{\texttt{a}}
    \end{tikzpicture}
  }

  Now we know that $\mathtt{a}$ occurs in the segment $(\lambda, \rho)$ in $u$
  but not in $v$, and thus we have established the situation illustrated in
  Fig.~\ref{fig:alt-ranker-char-middle}. Samson places his first pebble on
  an $\mathtt{a}$ within this section of $u$, and Delilah has to reply with
  a position outside of this section. No matter what side of the segment
  she chooses, with Lemma \ref{lem:alt-ranker-sides} Samson has a winning
  strategy for the remaining game and thus wins the \FOVDA{2}{n+1}{m} game.
  
  \picskip{3}
  In cases (3) and (5), we again assume that cases (2) and (4) do not apply, and
  we look at the same sets of rankers, $R_\ell$ and $R_r$, and
  at $r_n$, the $n$-prefix-ranker of $r$. We assume that $r(u) \le r'(u)$
  and that $r$ ends with $\rright$, all three other cases are completely
  symmetric. Notice that $r_n$ is an $(m-1,n)$-ranker, or an
  $(m,n)$-ranker that ends with $\rright$. Thus both structures agree on the
  ordering of $r_n$ and $r'$. The relative positions of all these rankers
  are illustrated in Fig.~\ref{fig:alt-ranker-char-closest}. As above, let
  $\lambda$ be the rightmost ranker from $R_\ell$ and let $\rho$ be the
  leftmost ranker from $R_r$, with respect to the ordering of these rankers
  on $v$. Again we know that $\lambda(u) < r(u) < \rho(u)$ and therefore the
  segment $(\lambda, \rho)$ of $u$ contains an $\mathtt{a}$. Notice that $r_n
  \in R_\ell$ and $r' \in R_r$, thus $r_n(v) \le \lambda(v) < \rho(v) \le
  r'(v)$. Thus the segment $(\lambda, \rho)$ does not contain the letter
  $\mathtt{a}$ in $v$, providing Samson with a winning strategy as argued above.

  \piccaption{\label{fig:alt-ranker-char-closest}Ranker positions, case (4)}
  \parpic(6.2cm,2.5cm)[fr]{
    \begin{tikzpicture}
      \word[5cm]{\wordu}{$u$}
      \word[5cm]{\wordv}{$v$}
      \wordupos{2cm}{$r$}
      \dwordpos{3cm}{$r'$}
      \wordvpos{4cm}{$r$}
      \dwordpos{1cm}{$r_n$}
    \end{tikzpicture}
  }

  To prove ``(i) $\Rightarrow$ (ii)'', we assume that the theorem holds for
  $n$, and that (i) holds for $(m,n+1)$, and we present a winning strategy
  for Delilah in the game $\FOVDA{2}{n+1}{m}(u,v)$ where Samson starts with
  a move on $u$.

  \picskip{3}
  If Samson places $x$ on a ranker $r \in R^\star_{m-1,n}(u)$, then Delilah
  replies by placing $x$ on the same ranker on $v$. Since (i)(b) holds for
  $(m,n+1)$, this establishes (i)(e) and (i)(f) for $(m,n)$. It also
  establishes (i)(e) and (i)(f) for $(m-1,n)$ with reversed roles of $u$ and
  $v$. Thus we can apply the inductive hypothesis to get a winning strategy
  for Delilah in the remaining game.

  If $x^u = y^u$ after Samson's first move, then Delilah replies with $x^v =
  y^v$. We use the inductive hypothesis to argue that Delilah wins the
  remaining $n$-move game, no matter what structure Samson chooses for his
  next move. If he chooses to play on $u$, then the remaining game is an
  $(m,n)$-game. Since in the first move Delilah set $x^v = y^v$, we have
  (i)(e) and (i)(f) for $(m,n)$, and thus the inductive hypothesis applies
  and Delilah wins the remaining game. On the other hand, if Samson chooses
  to play on $v$ for the next move, the remaining game is an $(m-1,n)$-game,
  since he started with a move on $u$. Because Delilah set $x^v = y^v$ in
  the first move, (i)(e) for $(m,n+1)$ implies both (i)(e) and (i)(f) for
  $(m-1,n)$ with reversed roles of $u$ and $v$. Thus we can again use the
  inductive hypothesis to get a winning strategy for Delilah in the
  remaining game.
  
  Otherwise we assume that $x^u < y^u$ after Samson's first move, the case
  for $x^u > y^u$ is completely symmetric. We look at the following two sets
  of rankers.
  \begin{align*}
    R_\ell &:= \{r \in R^\star_{m\rright, n}(u) \mid r(u) < x^u\}\\    
    R_r &:= \{r \in R^\star_{m\rleft, n}(u) \mid r(u) > x^u\}
  \end{align*}
  On $u$, all rankers from $R_\ell$ occur to the left of all rankers from
  $R_r$. Since (i)(c) holds for
  $(m,n+1)$, this is also true for the positions of these rankers on $v$.
  Let $\mathtt{a}$ be the letter Samson places his pebble on. To establish
  both (i)(e) and (i)(f) for $(m,n)$, Delilah needs to find an $\mathtt{a}$
  in $v$ that is to the right of all rankers from $R_\ell$ and to the left
  of all rankers from $R_r$. We define 
  \begin{align*}
    R^0_\ell &= \{r \in R^\star_{m\rright,n}(u) - R^\star_{m-1,n}(u) 
    \mid r(u) = x^u\}\\
    R^0_r &= \{r \in R^\star_{m\rleft,n}(u) - R^\star_{m-1,n}(u)
    \mid r(u) = x^u\}\\
    R'_\ell &:= \{r\rright_\mathtt{a} \mid r \in R_\ell\} \cup R^0_\ell
  \end{align*}
  and have Delilah place her pebble $x^v$ on the rightmost ranker from
  $R'_\ell$ on $v$. This position of course is labeled with an $\mathtt{a}$.
  Since on $u$ all rankers from $R_\ell'$ occur to the left of or at $x^u$,
  all of them occur strictly to the left of $y^u$. Since all rankers in
  $R'_\ell$ are from $R^\star_{m-1,n+1}(u)$ or $R^\star_{m\rright,n+1}(u)$,
  we can apply (i)(e) and (i)(f$_2$), and we see that all of these rankers
  also appear to the left of $y^v$. Therefore we have $x^v < y^v$, which
  makes sure that Delilah does not lose in this move, and also establishes
  (i)(d).
  
  To complete the inductive step, we need to argue that Delilah's move also
  establishes (i)(e) and (i)(f), both for $(m,n)$, and for $(m-1,n)$ with
  reversed roles of $u$ and $v$. Then, using the inductive hypothesis,
  Delilah has a winning strategy for the remaining game, no matter what side
  Samson chooses for his next move.

  We observe that all rankers from $R'_\ell$ appear to the right of the
  rankers from $R_r$. This is true by definition on $u$, and holds for $v$
  because (i)(b) and (i)(c) hold for $(m,n+1)$. Since Delilah placed $x^v$
  on a ranker from $R'_\ell$, we have (i)(e), (i)(f$_2$) and (i)(f$_4$) for
  $(m,n)$ for all all rankers from $R_r$. And since Delilah placed $x^v$ on
  the rightmost of the rankers from $R'_\ell$, we know that all rankers from
  $R_\ell$ appear to the left of $x^v$, just as they do on $u$. Thus we have
  (i)(e), (i)(f$_2$) and (i)(f$_4$) for the rankers from $R_\ell$ as well, and
  therefore for all rankers mentioned in those conditions.

  All rankers from $R^\star_{m\rright,n}$ that appear at $x^u$ are in
  $R^0_\ell$, since we already dealt with the case where $x^u$ does
  appear at a ranker from $R^\star_{m-1,n}$. Since Delilah chose $x^v$ as
  the rightmost ranker from $R'_\ell$, all of these rankers appear to the
  left of or at $x^v$, and we have established (i)(f$_1$) for $(m,n)$. For
  condition (i)(f$_3$), we need to argue about $R^0_r$. From (i)(b) and
  (i)(c) for $(m,n+1)$, we know that all rankers from $R^0_r$ appear to the
  right of or at the same position as the rankers from $R'_\ell$ on $v$,
  just as they do on $u$. Thus (i)(f$_3$) holds as well.

  Now that we have established (i) for $(m,n)$, we use the inductive
  hypothesis to get a winning strategy for Delilah for the remaining game if
  Samson's next move is on $u$. For the case where his next move is on $v$,
  we only need to establish (i) for $(m-1,n)$, but with reversed roles of
  $u$ and $v$. Reversing the roles of the two structures only affects
  condition (i)(f), and (i)(f) for $(m-1,n)$ follows immediately from (i)(e)
  for $(m,n)$. Thus Delilah also wins the remaining game if Samson's next
  move is on $v$.
\end{proof}
\end{full}

Using Theorem \ref{thm:alt-ranker-char}, we show that for any fixed alphabet
$\Sigma$, at most $|\Sigma| + 1$ alternations are useful. Intuitively, each
boundary position in a ranker says that a certain letter does not occur in
some part of a word. Alternations are only useful if they visit one of these
previous parts again. Once we visited one part of a word $|\Sigma|$ times,
this part cannot contain any more letters and thus is empty.

\begin{thm} \label{thm:alt-alphabet}
  Let $\Sigma$ be a finite alphabet, let $u, v \in \Sigma^\star$ and $n
  \in \N$. If $u \equiv^2_{|\Sigma|+1,n} v$, then $u \equiv^2_n v$.
\end{thm}

\begin{full}
\begin{proof}
  Suppose for the sake of a contradiction that $u \equiv^2_{|\Sigma|+1,n} v$
  and $u \not\equiv^2_n v$. Thus, using Theorem \ref{thm:alt-ranker-char},
  $u$ and $v$ agree on the definedness of all $(|\Sigma|+1,n)$-rankers, and
  on their order with respect to all $(|\Sigma|,n-1)$-rankers and some
  $(|\Sigma|+1,n-1)$-rankers. But since $u \not\equiv^2_n v$, $u$ and $v$
  need to disagree on the properties of some other ranker. Let $r = (p_1,
  \ldots, p_t)$ with $t \in \N$ be the shortest such ranker. We know that
  $r$ has more than $|\Sigma|$ blocks of alternating directions, say $r$ is
  an $m$-alternation ranker for some $m > |\Sigma|$. Let $1 \le k_1, \ldots,
  k_m \le t$ be the indices of the boundary positions at the end of each
  block, i.e. where $p_{k_i}$, $1 \le i < m$ points to a different direction
  than $p_{k_i+1}$. For the last of those indices we have $k_m = t$.

  We look at the prefix rankers of $r$ up to the
  end of each alternating block, $r_{k_i} := (p_1, \ldots, p_{k_i})$,
  and the
  intervals defined by these prefix rankers. We set $I_0(u) := [1, \abs{u}]$,
  $r_0(u) = 0$ if $p_1$ points to the right, and $r_0(u) = \abs{u}+1$ if $p_1$
  points to the left. For all $i \in [1, m]$ let,
  \begin{displaymath}
    I_i(u) :=
    \begin{cases}
      [r_{k_i-1}(u)+1, r_{k_i}(u)-1] & \text{if $p_{k_i}$ points to the right}\\
      [r_{k_i}(u)+1, r_{k_i-1}(u)-1] & \text{if $p_{k_i}$ points to the left}
    \end{cases}
  \end{displaymath}
  Notice that by definition the letter mentioned in $p_{k_i}$ does not occur
  in the interval $I_i$.

  Suppose that for all $i \in [1, m]$ we have $r_{k_i}(u) \in
  I_{i-1}(u)$.  Then the
  letter mentioned in $p_{k_i}$ has to occur in the interval $I_{i-1}(u)$ of $u$, but
  the interval $I_{|\Sigma|}(u)$ of $u$ cannot contain any of the $|\Sigma|$
  distinct letters. Therefore $r_{k_{|\Sigma|+1}} \notin I_{|\Sigma|}$ and
  we have a contradiction.

  Otherwise there is an $i \in [1,m]$ such that $r_{k_i}(u) \notin
  I_{i-1}(u)$. We will construct a ranker $r'$ that is shorter than $r$,
  does not have more alternations than $r$ and occurs at exactly the same
  position as $r$ in both $u$ and $v$. The main idea for this construction
  is that if $r_{k_i}(u) \notin I_{i-1}(u)$, then it is not useful to enter
  this interval at all. By our assumption, $u$ and $v$
  disagree on some property of the ranker $r$, and thus on some property of
  the shorter ranker $r'$. This contradicts our assumption that $r$ was the
  shortest such ranker.

  Now we show how to construct a shorter ranker $r'$ that occurs at the same
  position as $r$. We assume
  without loss of generality that $p_{k_i}$ points to the left. In this case
  we have 
  $r_{k_i}(u) \notin I_{i-1}(u) = [r_{k_{i-1}-1}(u)+1, r_{k_{i-1}}(u)-1]$.
  We look at the
  relative positions of the rankers $r_{k_{i-1}+1}, \ldots, r_{k_i}$ with
  respect to the ranker $r_{k_{i-1}-1}$. We know that $r_{k_i}(u) \le
  r_{k_{i-1}-1}(u)$, and we are interested in the right-most of the rankers
  $r_{k_{i-1}+1}, \ldots, r_{k_i}$ that is still outside of the interval
  $I_{i-1}(u)$. Let $r_j$ be this ranker.
  Thus we have
  \begin{displaymath}
    r_{k_i}(u) < \ldots < r_j(u) \le r_{k_{i-1}-1}(u) < r_{j-1}(u) < \ldots <
    r_{k_{i-1}+1}(u) < r_{k_{i-1}}(u)
  \end{displaymath}
  We know that $u \equiv^2_{|\Sigma|+1,n} v$, thus by Theorem
  \ref{thm:alt-ranker-char}, these rankers occur in exactly the same order
  in $v$. Now we set $s := (r_{k_{i-1}-1}, p_j, \ldots, p_{k_i})$.
  Because $u$ and $v$ agree on the ordering of the relevant rankers, we have
  $s(u) = r_{k_i}(u)$ and $s(v) = r_{k_i}(v)$. Therefore we have reduced the
  size of a prefix of $r$ without increasing the number of alternations, and
  thus have a shorter ranker $r'$ that occurs at the same position as $r$ in
  both structures.\qedconf
\end{proof}
\end{full}

In order to prove that the alternation hierarchy for $\FOV{2}$ is strict,
we define example languages that can be separated by a formula of a
given alternation depth $m$, but that cannot be separated by any formula of
lower alternation depth. As Theorem \ref{thm:alt-alphabet} shows,
we need to increase
the size of the alphabet with increasing alternation depth.
We inductively define the example words $u_{m,n}$ and $v_{m,n}$ and the
example languages $K_m$ and $L_m$ over finite alphabets $\Sigma_m = \{\mathtt{a}_0,
\ldots, \mathtt{a}_{m-1}\}$. Here $i$, $m$ and $n$ are positive integers.
\begin{align*}
  u_{1,n} &:= \mathtt{a}_0 & v_{1,n} &:= \varepsilon\\
  u_{2,n} &:= \mathtt{a}_0 (\mathtt{a}_1\mathtt{a}_0)^{2n} & v_{2,n} &:= (\mathtt{a}_1\mathtt{a}_0)^{2n}\\
  u_{2i+1,n} &:= (\mathtt{a}_0 \ldots \mathtt{a}_{2i})^{n} \; u_{2i,n}
  & v_{2i+1,n} &:= (\mathtt{a}_0 \ldots \mathtt{a}_{2i})^{n} \; v_{2i,n}\\
  u_{2i+2,n} &:= u_{2i+1,n} \; (\mathtt{a}_{2i+1} \ldots \mathtt{a}_0)^{n}
  & v_{2i+2,n} &:= v_{2i+1,n} \; (\mathtt{a}_{2i+1} \ldots \mathtt{a}_0)^{n}
\end{align*}
Notice that $u_{m,n}$ and $v_{m,n}$ are almost identical -- if we delete
only one $\mathtt{a}_0$ from $u_{m, n}$, we get $v_{m,n}$. Finally, we set $K_m :=
\bigcup_{n \ge 1} \{u_{m,n}\}$ and $L_m := \bigcup_{n \ge 1} \{v_{m,n}\}$.

\begin{defi}
  A formula $\varphi$ \emph{separates} two languages $K, L \subseteq
  \Sigma^\star$ if for all $w \in K$ we have $w \models \varphi$ and for all
  $w \in L$ we have $w \not\models \varphi$ or vice versa.
\end{defi}

\begin{full}

\begin{lem} 
  \label{lem:alt-hierarchy-exp}
  For all $m \in \N$, there is a formula $\varphi_m \in \FOVA[$<$]{2}{m}$ that
  separates $K_m$ and $L_m$. 
\end{lem}

\begin{proof}
  For $m = 1$, we can easily separate $K_1 = \{\mathtt{a}_0\}$ and $L_1 =
  \{\varepsilon\}$ with the formula $\exists x (x=x)$.
  For all larger $m$, we show that the two languages $K_m$ and $L_m$ differ on
  the ordering of two $(m-1)$-alternation rankers. Then by Theorem
  \ref{thm:alt-ranker-char} there is an \FOVDA[$<$]{2}{m}{m} formula that
  separates $K_m$ and $L_m$. We inductively define the rankers
  \begin{align*}
    r_2 &:= \rright_{\mathtt{a}_0}
    & s_2 &:= \rright_{\mathtt{a}_1}\\
    r_{2i+1} &:= \rleft_{\mathtt{a}_{2i}} r_{2i}
    &s_{2i+1} &:= \rleft_{\mathtt{a}_{2i}} s_{2i}\\
    r_{2i+2} &:= \rright_{\mathtt{a}_{2i+1}} r_{2i+1}
    & s_{2i+2} &:= \rright_{\mathtt{a}_{2i+1}} s_{2i+1}
  \end{align*}

  For $m = 2$, it is easy to see that $r_2(u_{2,n}) < s_2(u_{2,n})$, but
  $r_2(v_{2,n}) > s_2(v_{2,n})$. For $m > 2$, these rankers disagree on
  their order as well. To prove this, we prove the following two equalities.
  \begin{equation*}
    r_{2i+2}(u_{2i+2,n}) = r_{2i+1}(u_{2i+1,n}) = (2i+1)n + r_{2i}(u_{2i,n})
  \end{equation*}
  To prove this, we first use the definitions above and write
  \begin{equation*}
    r_{2i+2}(u_{2i+2,n}) = (\rright_{\mathtt{a}_{2i+1}} r_{2i+1})
    (u_{2i+1,n} \; (\mathtt{a}_{2i+1} \ldots \mathtt{a}_0)^{n})
  \end{equation*}
  The letter $\mathtt{a}_{2i+1}$ does not occur in the word $u_{2i+1,n}$,
  and thus $\rright_{\mathtt{a}_{2i+1}}(u_{2i+2,n})$ points to the first
  position in $u_{2i+2,n}$ right after the copy of $u_{2i+1,n}$. We observe
  that $r_{2i+1}$ starts with $\rleft$, and that $r_{2i+1}$ is defined on
  $u_{2i+1,n}$. Thus the evaluation of the remainder of $r_{2i+2}$ on
  $u_{2i+2,n}$ never leaves the copy of $u_{2i+1,n}$, and we have
  \begin{equation*}
    r_{2i+2}(u_{2i+2,n}) = r_{2i+1}(u_{2i+1,n})
  \end{equation*}
  For the second part of the equality, we have
  \begin{equation*}
    r_{2i+1} (u_{2i+1,n})
    = (\rleft_{\mathtt{a}_{2i}} r_{2i}) (
    (\mathtt{a}_0 \ldots \mathtt{a}_{2i})^{n} \; u_{2i,n})\\
  \end{equation*}
  As above, the letter $\mathtt{a}_{2i}$ does not occur in the word
  $u_{2i,n}$, and thus $\rleft_{\mathtt{a}_{2i}}(u_{2i+1,n})$ points to the
  position in $u_{2i+1,n}$ right before the copy of $u_{2i,n}$. The ranker
  $r_{2i}$ starts with $\rright$, and $r_{2i}$ is defined on $u_{2i,n}$.
  Thus, just as above, the evaluation of the remainder of $r_{2i+1}$ on
  $u_{2i+1,n}$ never leaves the copy of $u_{2i,n}$, and we have
  \begin{equation*}
    r_{2i+1}(u_{2i+1,n}) = (2i+1)n + r_{2i}(u_{2i,n})
  \end{equation*}
  Exactly the same holds for the other rankers ($s_2, \ldots$) and words
  $(v_{2,n}, \ldots$). We have
  \begin{align*}
    r_{2i+2}(u_{2i+2,n}) = r_{2i+1}(u_{2i+1,n}) = (2i+1)n + r_{2i}(u_{2i,n})\\
    s_{2i+2}(u_{2i+2,n}) = s_{2i+1}(u_{2i+1,n}) = (2i+1)n + s_{2i}(u_{2i,n})\\
    r_{2i+2}(v_{2i+2,n}) = r_{2i+1}(v_{2i+1,n}) = (2i+1)n + r_{2i}(v_{2i,n})\\
    s_{2i+2}(v_{2i+2,n}) = s_{2i+1}(v_{2i+1,n}) = (2i+1)n + s_{2i}(v_{2i,n})
  \end{align*}

  Now an easy inductive argument, based on the two equalities we just
  proved, shows that the rankers disagree on their order. Therefore
  condition (i)(b) of Theorem \ref{thm:alt-ranker-char} fails for any pair
  of words, and there is a formula in \FOVDA[$<$]{2}{m}{m} that separates
  $K_m$ and $L_m$.\qedconf
\end{proof}\pagebreak[2]

\begin{lem} 
  \label{lem:alt-hierarchy-nexp}
  For $m \in \N$, $m \ge 1$,
  and all $n \in \N$, we have $u_{m,n} \equiv^2_{m-1,n} v_{m,n}$.
\end{lem}

\begin{proof}
  Because we do not have constants, there are no quantifier-free sentences.
  Thus $\FOVDA[$<$]{2}{n}{0}$ 
  does not contain any formulas and the statement holds
  trivially for $m = 1$.

  For $m \ge 2$ and any $n \ge m$, we claim that exactly the same
  $(m-1,n)$-rankers are defined over $u_{m,n}$ and $v_{m,n}$, and that all
  $(m-1,n)$-rankers appear in the same order with respect to all
  $(m-2,n-1)$-rankers and all $(m-1,n-1)$-rankers that end on a different
  direction. Once we established this claim, the lemma follows immediately
  from Theorem \ref{thm:alt-ranker-char}. We already observed that $u_{m,n}$
  and $v_{m,n}$ are almost identical. The only difference between the two
  words is that $u_{m,n}$ contains the letter $\mathtt{a}_0$ in the middle whereas
  $v_{m,n}$ does not. Thus we only have to consider rankers that are
  affected by this middle $\mathtt{a}_0$.

  We claim that any ranker that points to the middle $\mathtt{a}_0$ of $u_{m,n}$
  requires at least $m-1$ alternations. Furthermore, we claim that any such
  ranker needs to start with $\rright$ for even $m$ and with $\rleft$ for
  odd $m$. We prove this by induction on $m$.

  For $m=2$ we have $u_{2,n} = \mathtt{a}_0 (\mathtt{a}_1\mathtt{a}_0)^n$. Any $n$-ranker that starts
  with $\rleft$ cannot reach the first $\mathtt{a}_0$, thus we need a ranker that
  starts with $\rright$.

  For odd $m > 2$ we have $u_{m,n} = (\mathtt{a}_0 \ldots \mathtt{a}_{m-1})^n u_{m-1,n}$. Any
  $n$-ranker that starts with $\rright$ cannot leave the first block of $n
  \cdot m$ symbols of this word and thus not reach the middle $\mathtt{a}_0$.
  Therefore we need to start with $\rleft$, and in fact use
  $\rleft_{\mathtt{a}_{m-1}}$ at some point, because we would not be able to leave
  the last section of $u_{m-1,n}$ otherwise. But with $\rleft_{\mathtt{a}_{m-1}}$ we
  move past all of $u_{m-1,n}$, and we need one alternation to turn around
  again. By induction, we need at least $m-2$ alternations within
  $u_{m-1,n}$, and thus $m-1$ alternations total.

  The argument for even $m$ is completely symmetric. Thus we showed that we
  need at least $m-1$ alternation blocks to point to the middle $\mathtt{a}_0$.
  Furthermore, we showed that if we have exactly $m-1$ alternation blocks,
  then the last of these blocks uses $\rright$.
  Therefore we only need to consider
  $(m-1)$-alternation rankers that end on $\rright$ and pass through the
  middle $\mathtt{a}_0$. It is easy to see that all of these rankers agree on their
  ordering with respect to all other $(m-2)$-alternation rankers, and with
  respect to all $(m-1)$-alternation rankers that end on $\rleft$.

  To summarize, we showed that $u_{m,n}$ and $v_{m,n}$ satisfy condition (i)
  from Theorem \ref{thm:alt-ranker-char} for $m-1$ alternations. Thus the
  two words agree on all formulas from \FOVDA[$<$]{2}{n}{m-1}.\qedconf
\end{proof}
\end{full}

\begin{conference}
  Our example words are constructed such that for $m \ge 3$, $u_{m,n}$ and
  $v_{m,n}$ can be distinguished by the ordering of two $(m-1)$-rankers. In
  the case $m=3$ for example, we can use the two rankers $r_3 :=
  \rleft_{\mathtt{a}_2} \rright_{\mathtt{a}_0}$ and $s_3 := \rleft_{\mathtt{a}_2} \rright_{\mathtt{a}_1}$. A
  formal argument for all $m$ is given in \cite{WI07}. There we also argue
  that the example words $u_{m,n}$ and $v_{m,n}$
  agree on the definedness of all $(m-1,n)$-rankers, and that these rankers
  appear in exactly the same order with respect to shorter rankers. Thus
  the two languages $K_m$ and $L_m$ cannot be separated by any
  $\FOVA[$<$]{2}{m-1}$ formula. Thus we have the following theorem.
\end{conference}

\begin{thm}[alternation hierarchy for {\FOV[$<$]{2}}]
  \label{thm:alt-hierarchy}
  For any positive integer $m$, there is a $\varphi_m \in \FOVA[$<$]{2}{m}$
  and there are two languages $K_m, L_m$ such that $\varphi_m$ separates
  $K_m$ and $L_m$, but no $\psi \in \FOVA[$<$]{2}{m-1}$ separates $K_m$ and
  $L_m$.
\end{thm}

\begin{full}
\begin{proof}
  The theorem immediately follows from Lemma \ref{lem:alt-hierarchy-exp}
  and Lemma \ref{lem:alt-hierarchy-nexp}.\qedconf
\end{proof}
\end{full}

Theorem \ref{thm:alt-hierarchy} resolves an open question from
\cite{EVW97,EVW02}.

\section{Structure Theorem and Alternation Hierarchy for \texorpdfstring{\FOV[$<,\suc$]{2}}{FOV[<,suc]{2}}}
\label{sec:FO2-suc}

We extend our definitions of boundary positions and rankers from 
Sect.~\ref{sec:FO2} to include the substrings of a given length that occur
immediately before and after the position of the ranker.

\begin{defi}
  A \emph{$(k,\ell)$-neighborhood boundary position} denotes the first or
  last occurrence of a substring in a word. More precisely, a
  $(k,\ell)$-neighborhood boundary position is of the form $d_{(s, \mathtt{a}, t)}$
  with $d \in \{\rright, \rleft\}$, $s \in \Sigma^k$, $\mathtt{a} \in \Sigma$ and $t
  \in \Sigma^\ell$. The interpretation of a $(k,\ell)$-neighborhood boundary
  position $p = d_{(s,\mathtt{a},t)}$ on a word $w = w_1 \ldots w_{|w|}$ is defined
  as follows.
  \begin{displaymath}
    p(w) =
    \begin{cases}
      \min\{i \in [k+1, |w|-\ell] \mid w_{i-k} \ldots w_{i+\ell}
      = s \, \mathtt{a} \, t\}
      & \text{if } d = \rright\\
      \max\{i \in [k+1, |w|-\ell] \mid w_{i-k} \ldots w_{i+\ell}
      = s \, \mathtt{a} \, t\}
      & \text{if } d = \rleft
    \end{cases}
  \end{displaymath}
  Notice that $p(w)$ is undefined if the sequence $s \mathtt{a} t$
  does not occur in
  $w$. A $(k,\ell)$-neighborhood boundary position can also be specified
  with respect to a position $q \in [1, |w|]$.
  \begin{displaymath}
    p(w, q) =
    \begin{cases}
      \min\{i \in [\max\{q+1,k+1\}, |w|-\ell]
      \mid w_{i-k} \ldots w_{i+\ell} = s \, \mathtt{a} \, t\}
      & \text{if } d = \rright\\
      \max\{i \in [k+1, \min\{q-1, |w|-\ell\}] 
      \mid w_{i-k} \ldots w_{i+\ell} = s \, \mathtt{a} \, t\}
      & \text{if } d = \rleft
    \end{cases}
  \end{displaymath}
\end{defi}

Observe that $(0,0)$-neighborhood boundary positions are identical to the
boundary positions from Definition \ref{def:boundary-position}. As before in
the case without successor, we build rankers out of these boundary
positions. The size of the boundary position neighborhoods grows linearly
from the first boundary position to the last one, reflecting the remaining
quantifier depth for successor moves at those positions.

\begin{defi}
  An \emph{$n$-successor-ranker} $r$ is a sequence of $n$ neighborhood
  boundary positions, $r = (p_1, \ldots, p_n)$, where $p_i$ is a
  $(k_i,\ell_i)$-neighborhood boundary position and $k_i, \ell_i \in
  [0,i-1]$. The interpretation of an $n$-successor-ranker $r$ on a word
  $w$ is defined as follows.
  \begin{displaymath}
    r(w) :=
    \begin{cases}
      p_1(w) & \text{if } r = (p_1)\\
      \text{undefined} & \text{if $(p_1, \ldots, p_{n-1})(w)$ is
        undefined}\\
      p_n(w, (p_1, \ldots, p_{n-1})(w)) & \text{otherwise}
    \end{cases}
  \end{displaymath}
  We denote the set of all $n$-successor-rankers that are defined over a
  word $w$ by $\mathit{SR}_n(w)$, and set $\mathit{SR}_n^\star(w) :=
  \bigcup_{i \in [1, n]} \mathit{SR}_i(w)$.
\end{defi}

Because we now have the additional atomic relation $\suc$, we need to extend
our definition of order type as well.

\begin{defi} \label{def:suc-order-type} Let $i, j \in \N$. The
  \emph{successor order type} of $i$ and $j$ is defined as
  \begin{displaymath}
    \sucord(i,j) =
    \begin{cases}
      \ll & \text{if } i < j-1\\
      -1 & \text{if } i = j-1\\
      = & \text{if } i = j\\
      +1 & \text{if } i = j+1\\
      \gg & \text{if } i > j+1\\
    \end{cases}
  \end{displaymath}
\end{defi}

With this new definition of $n$-successor-rankers, our proofs for Lemmas
\ref{lem:ranker-sides}, \ref{lem:ranker-def}, \ref{lem:ranker-exp} and
Theorem \ref{thm:ranker-char} go through with only minor modifications.
Instead of working through all the details again, we simply point out the
differences.

First we notice that $1$-successor-rankers are simply $1$-rankers, so the
base case of all inductions remains unchanged. In the proofs of Lemmas
\ref{lem:ranker-sides}, \ref{lem:ranker-def} and \ref{lem:ranker-exp}, and
in the proof of ``(ii) $\Rightarrow$ (i)'' from Theorem \ref{thm:ranker-char},
we argued that Delilah cannot reply with a position in a given section
because it does not contain a certain ranker and therefore it does not
contain the symbol used to define this ranker. Now we need to know more --
we need to show that Delilah cannot reply with a certain letter in a given
section that is surrounded by a specified neighborhood, given that this
section does not contain the corresponding successor-ranker. Whenever
Samson's winning strategy depends on the fact that an $n$-successor-ranker
does not occur in a given section, he has $n-1$ additional moves left. So if
Delilah does not reply with a position with the same letter and the same
neighborhood, Samson can point out a difference in the neighborhood with at
most $(n-1)$ additional moves.

For the other direction of Theorem \ref{thm:ranker-char}, we need to make
sure that Delilah can reply with a position that is contained in the correct
interval, has the same symbol and is surrounded by the same neighborhood.
Where we previously defined the $n$-ranker $s := (\lambda,
\rright_\mathtt{a})$ or $s := (\rho, \rleft_\mathtt{a})$, we now include the
$(n-1)$-neighborhood of the respective positions chosen by Samson. Thus we
make sure that Samson cannot point out a difference in the two words, and
Delilah still has a winning strategy. Thus we have the following three
theorems for \FOV[$<,\suc$]{2}.

\begin{thm}[structure of {\FOVD[$<,\suc$]{2}{n}}]
  \label{thm:suc-structure}
  Let $u$ and $v$ be finite words, and let $n \in \N$. The following two
  conditions are equivalent.
  \begin{enumerate}[\em(i)]
  \item
    \begin{enumerate}[\em(a)]
    \item $\mathit{SR}_n(u) = \mathit{SR}_n(v)$, and,
    \item for all $r \in \mathit{SR}_n^\star(u)$ and for all $r' \in
      \mathit{SR}_{n-1}^\star(u)$,\\
      $\sucord(r(u), r'(u)) = \sucord(r(v), r'(v))$
    \end{enumerate}
  \item $u \equiv^2_n v$
  \end{enumerate}
\end{thm}

\begin{thm}[structure of {\FOVDA[$<,\suc$]{2}{n}{m}}]
  \label{thm:suc-alt-structure} 
  Let $u$ and $v$ be finite words, and let $m, n \in \N$ with $m \le n$. The
  following two conditions are equivalent.
  \begin{enumerate}[\em(i)]
  \item
    \begin{enumerate}[\em(a)]
    \item $\mathit{SR}_{m,n}(u) = \mathit{SR}_{m,n}(v)$, and,
    \item for all $r \in \mathit{SR}_{m,n}^\star(u)$ and for all $r' \in
      \mathit{SR}_{m-1,n-1}^\star(u)$,\\
      $\sucord(r(u),r'(u)) = \sucord(r(v), r'(v))$, and,
    \item for all $r \in \mathit{SR}_{m,n}^\star(u)$ and $r' \in
      \mathit{SR}_{m,n-1}^\star(u)$ such that $r$ and $r'$ end with
      different directions, $\sucord(r(u),r'(u)) = \sucord(r(v), r'(v))$
    \end{enumerate}
  \item $u \equiv^2_{m,n} v$
  \end{enumerate}
\end{thm}

\begin{thm}[alternation hierarchy for {\FOV[$<,\suc$]{2}}]
  \label{thm:suc-alt-hierarchy}
  Let $m$ be a positive integer. There is a
  $\varphi_m \in \FOVA[$<,\suc$]{2}{m}$ and there are two languages $K_m, L_m
  \subseteq \Sigma^\star$ such that $\varphi_m$ separates $K_m$ and $L_m$,
  but there is no $\psi \in \FOVA[$<,\suc$]{2}{m-1}$ that separates $K_m$ and
  $L_m$.
\end{thm}

\begin{full}
\begin{proof}
  We use the same ideas as before in Theorem \ref{thm:alt-hierarchy}. We
  define example languages that now include an extra letter \texttt{b} to ensure
  that the successor predicate is of no use. As before, we inductively
  construct the words $u_{m,n}$ and $v_{m,n}$ and use them to define the
  languages $K_m$ and $L_m$.

  \begin{align*} u_{1,n} &:= \mathtt{b}^{2n} \mathtt{a}_0
    \mathtt{b}^{2n} & v_{1,n} &:= \mathtt{b}^{2n}\\ u_{2,n} &:=
    u_{1,n} \;
    (\mathtt{a}_1\mathtt{b}^{2n}\mathtt{a}_0\mathtt{b}^{2n})^{2n} &
    v_{2,n} &:= v_{1,n} \;
    (\mathtt{a}_1\mathtt{b}^{2n}\mathtt{a}_0\mathtt{b}^{2n})^{2n}\\
    u_{2i+1,n} &:= (\mathtt{b}^{2n} \mathtt{a}_0 \mathtt{b}^{2n}
    \ldots \mathtt{b}^{2n} \mathtt{a}_{2i})^n \; u_{2i,n} & v_{2i+1,n}
    &:= (\mathtt{b}^{2n} \mathtt{a}_0 \mathtt{b}^{2n} \ldots
    \mathtt{b}^{2n} \mathtt{a}_{2i})^n \; v_{2i,n}\\ u_{2i+2,n} &:=
    u_{2i+1,n} \; (\mathtt{a}_{2i+1} \mathtt{b}^{2n} \ldots
    \mathtt{a}_0 \mathtt{b}^{2n})^{n} & v_{2i+2,n} &:= v_{2i+1,n} \;
    (\mathtt{a}_{2i+1} \mathtt{b}^{2n} \ldots \mathtt{a}_0
    \mathtt{b}^{2n})^{n} \end{align*} Finally we set $K_m :=
    \bigcup_{n \ge 1} \{u_{m,n}\}$ and $L_m := \bigcup_{n \ge 1}
    \{v_{m,n}\}$. Notice that the \texttt{b}s are not necessary to
    distinguish between the two languages $K_m$ and $L_m$, and thus
    the proof of Lemma \ref{lem:alt-hierarchy-exp} goes through
    unchanged and we have a formula $\varphi_m \in
    \FOVA[$<,\suc$]{2}{m}$ that separates $K_m$ and $L_m$. To see that
    no $\FOVA[$<,\suc$]{2}{m-1}$ formula can separate $K_m$ and $L_m$,
    we observe that any $(n-1)$-neighborhood in the words $u_{m,n}$
    and $v_{m,n}$ contains all \texttt{b}s except for at most one
    letter $\mathtt{a}_i$ for some $i \in [0,m-1]$. Thus the proof of
    Lemma \ref{lem:alt-hierarchy-nexp} goes through here as
    well.\qedconf
\end{proof}
\end{full}

\begin{conference}
  The proof of Theorem \ref{thm:suc-alt-hierarchy} is given in \cite{WI07},
  and mostly identical to the proof of Theorem \ref{thm:alt-hierarchy}. We use
  $n$ copies of an extra letter between any two letters in our example words,
  and thus ensure that the successor predicate is not useful.
\end{conference}

\section{Small Models and Satisfiability for \texorpdfstring{$\FOV{2}[<]$}{FOV2[<]}}
\label{sec:satisfiability}

\newcommand{\TILING}{\lang{TILING}}

The complexity of satisfiability for $\FOV{2}[<]$ was investigated in
\cite{EVW02}. There it is shown that any satisfiable $\FOVD{2}{n}[<]$
formula has a model of size at most exponential in $n$. It follows that
satisfiability for $\FOV{2}[<]$ is in \NEXP, and a reduction from \TILING{}
shows that satisfiability for $\FOV{2}[<]$ is \NEXP-complete. Using our
characterization of $\FOV{2}[<]$, Wilke
observed that satisfiability becomes $\NP$-complete if we look at binary
alphabets only \cite{W07}. We generalize this observation and show that
satisfiability for $\FOV{2}[<]$ is $\NP$-complete for any fixed alphabet
size. In contrast to this, satisfiability for $\FOV{2}[<,\suc]$ is
$\NEXP$-complete even for binary alphabets \cite{EVW02}, since in the
presence of a successor predicate we can encode an arbitrary alphabet in
binary. Before we state and prove the two theorems of this section, we
prove a simple technical lemma first.

\begin{lem}
  \label{lem:equiv-word-replace}
  Let $u, v, v', w \in \Sigma^\star$. If $v \equiv^2_n v'$, then $uvw \equiv^2_n
  uv'w$.
\end{lem}

\begin{proof}
  We argue that Delilah has a winning strategy for the game
  $\FOVD{2}{n}(uvw, uv'w)$: If Samson places a pebble in $u$ or $w$, Delilah
  replies with the identical position in $u$ or $w$ in the other structure.
  If Samson places a pebble in $v$ or $v'$, Delilah replies according to her
  winning strategy in the game $\FOVD{2}{n}(v,v')$. All of these moves
  obviously preserve the ordering of the pebbles, and thus Delilah wins.
\end{proof}

\begin{thm}[Small Model Property for Bounded Alphabets]
  \label{thm:small-model}
  Let $n \in \N$ and let $\varphi \in
  \FOVD{2}{n}[<]$ be a formula over a $k$-letter alphabet. If $\varphi$ is
  satisfiable, then $\varphi$ has a model of size $O(n^k)$.
\end{thm}

\begin{conference}
  The proof of Theorem \ref{thm:small-model} is presented in \cite{WI07}. We
  argue that any fixed word has as most $O(n^k)$ positions that can be
  reached with $n$-rankers, and thus we have a word of size $O(n^k)$ that
  satisfies the given formula.
\end{conference}

\begin{full}
\begin{proof}
  Let $w$ be an arbitrary model of $\varphi$. We use induction on $k$ to
  show how to construct a new model of size $O(n^k)$ that satisfies
  $\varphi$. For $k=1$, i.e. a single letter alphabet, 
  we observe that an $n$-ranker can only point to a
  position within the first or last $n$ letters of $w$. We let $w'$ be a
  copy of $w$ with all letters after the first $n$ letters and before the
  last $n$ letters removed. The words $w$ and $w'$ agree on the existence
  and ordering of all $n$-rankers, thus we can apply Theorem
  \ref{thm:ranker-char} and it follows that $w' \models \varphi$.

  For the inductive case, we partition $w$ into segments, where each segment
  is a maximal sequence of the same letter. For example, the word
  $\mathtt{aaabb}$ has two segments, $\mathtt{aaa}$ and $\mathtt{bb}$.
  First, we let $w'$ be a copy of $w$ where we cut down all segments that
  are longer than $2n$ to exactly $2n$ letters. Since no $n$-ranker can
  point to a position within any segment after the first $n$ letters and
  before the last $n$ letters of that segment, we have $w' \models \varphi$.

  Now we partition the word $w'$ such that $w' = u_1 s_1 u_2 \ldots u_r s_r
  u_{r+1}$, where $r \in \N$ and for every $1 \le i \le r$, $u_i$ is a
  string of maximal length that uses exactly $k$ different letters, $s_i$ is
  a segment, and $u_{r+1}$ is a string over at most a $k$-letter alphabet.
  We observe that this partition is unique: If $\mathtt{a}$ is the last
  of the $(k+1)$ letters in our alphabet to appear in $w'$, starting from
  the left, then $s_1$ is the left-most segment of $\mathtt{a}$'s, and $u_1$
  is everything up to that segment. Now $s_2$ is the left-most segment after
  $s_1$ of the letter that appears last after $s_1$, and so on. We can point
  to a position in segment $s_n$ with an $n$-ranker, but no $n$-ranker that
  starts with $\rright$ can point to a position to the right of $s_n$.
  Similarly, we partition $w'$, now starting from the right, such that $w' =
  v_{q+1} t_q v_q \ldots v_2 t_1 v_1$, where $q \in \N$ and for every $1 \le
  i \le q$, $v_i$ is a string of maximal length that uses exactly $k$
  different letters, $t_i$ is a segment, and $v_{q+1}$ is a string over at
  most a $k$-letter alphabet. Again, this partition is unique and any
  $n$-ranker that starts with $\rleft$ cannot point to a position to the
  left of $t_n$. We also notice that both partitions have the same number
  of segments, i.e. $r = q$, since any substring $u_i s_i$ from the first
  partition contains all letters of the alphabet and thus has to contain
  at least one segment $t_j$ from the second partition, and vice versa.

  If both partitions use more than $2n$ segments, then the segment
  $s_n$ of the first partition occurs to the left of the segment $t_n$ of
  the second partition. In this case we construct the word $w'' = u_1 s_1
  u_2 \ldots u_n s_n t_n v_n \ldots v_2 t_1 v_1$. $w''$ agrees with $w'$ on
  all $n$-rankers, and thus $w'' \models \varphi$. Every one of the strings
  $u_1, \ldots, u_n$ and $v_1, \ldots v_n$ uses at most $k$ different
  letters, therefore we can apply the inductive hypothesis and replace each
  of these strings with an equivalent string of length $O(n^k)$, as
  explained in Lemma \ref{lem:equiv-word-replace}. Thus we
  have constructed a word of length $O(n^{k+1})$ that satisfies $\varphi$.

  If the partitions have at most $2n$ segments, then we combine the
  two partitions such that $w' = w_1 x_1 \ldots x_p w_{p+1}$, where $p
  \le 4n$, and for every $1 \le i \le p$, $x_p$ is one of the original
  segments $s_1, \ldots, s_r$ and $t_1, \ldots, t_q$. As above, we use the
  inductive hypothesis to replace all strings $x_i$ with equivalent strings
  of length $O(n^k)$, and thus construct a new string of length $O(n^{k+1})$
  that satisfies $\varphi$.\qedconf
\end{proof}
\end{full}

\begin{thm}
  Satisfiability for $\FOV{2}[<]$ where the size of the alphabet is bounded
  by some fixed $k \ge 2$ is \NP-complete.
\end{thm}

\begin{proof}
  Membership in $\NP$ follows immediately from Theorem \ref{thm:small-model}
  -- we nondeterministically guess a model of size $O(n^k)$ where $n$ is the
  quantifier depth of the given formula, and verify that it is a model of
  the formula. Now we give a reduction from $\SAT$. Let $\alpha$ be a
  boolean formula in conjunctive normal form over the variables $X_1,
  \ldots, X_n$. We construct a $\FOV{2}[<]$ formula $\varphi = \varphi_n
  \land \alpha[\xi_i / X_i]$, where $\varphi_n$ says that every model has
  size exactly $n$, and where we replace every occurrence of $X_i$ in
  $\alpha$ with a formula $\xi_i$ of length $O(n)$ which says that the
  $i$-th letter is a $1$. The total length of $\varphi$ is $O(|\alpha| \cdot
  n)$, and $\varphi$ is satisfiable iff $\alpha$ is satisfiable.\qedconf
\end{proof}

\section{Conclusion}

We proved precise structure theorems for \FOV{2}, with and without the
successor predicate, that completely characterize the expressive power of
the respective logics, including exact bounds on the quantifier depth and on
the alternation depth. Using our structure theorems, we showed that the
quantifier alternation hierarchy for \FOV{2} is strict, settling an open
question from \cite{EVW97,EVW02}. Both our structure theorems and the
alternation hierarchy results add further insight to and simplify previous
characterizations of \FOV{2}. We hope that the insights gained in
our study of $\fo^2$ on words will be useful in
future investigations of the trade-off between formula size and number of
variables.

\section*{Acknowledgment}

We would like to thank Thomas Wilke for pointing out the consequences of our
structural results to the satisfiability problem for $\FOV{2}[<]$. We are
also very thankful to two anonymous reviewers, whose detailed comments and
suggestions significantly improved the presentation of our results.

\vskip-40 pt
\end{document}